%% file: iclr2026_conference.tex
\documentclass{article} % For LaTeX2e
\usepackage{iclr2026_conference,times}
\input{math_commands.tex}

\usepackage{hyperref}
\usepackage{url}
\usepackage{graphicx}
\usepackage{subfigure}
\usepackage{booktabs} 
\usepackage{amsfonts}       
\usepackage{nicefrac}      
\usepackage{microtype}      
\usepackage{xcolor}       
\usepackage{amsmath}
\usepackage{amssymb}
\usepackage{mathtools}
\usepackage{amsthm}
\usepackage{wrapfig}
\usepackage{algorithm}
\usepackage{algorithmic}
\usepackage{enumitem}
\usepackage{multirow}
\usepackage[capitalize,noabbrev]{cleveref}
\usepackage{marvosym}

\theoremstyle{plain}
\ifx\theorem\undefined
\newtheorem{theorem}{Theorem}

\newtheorem{lemma}{Lemma}
\newtheorem{proposition}{Proposition}

\theoremstyle{definition}
\newtheorem{definition}{Definition}
\newtheorem{assumption}{Assumption}
\newtheorem{remark}{Remark}

\title{Condition Errors Refinement in Autoregressive Image Generation with Diffusion Loss}

\author{Yucheng Zhou\thanks{Equal Contribution.}~~, ~Hao Li\footnotemark[1]~~, ~Jianbing Shen$^{\text{\Letter}}$ \\
SKL-IOTSC, CIS, University of Macau \\
\texttt{yucheng.zhou@connect.um.edu.mo}
}

\iclrfinalcopy 
\begin{document}
\maketitle
\renewcommand{\thefootnote}{\Letter} 
\footnotetext{Corresponding Author.}
\renewcommand{\thefootnote}{\arabic{footnote}} 
\input{main}

\bibliography{ref}
\bibliographystyle{iclr2026_conference}
\input{appendix}

\end{document}

%% file: math_commands.tex
\usepackage{amsmath,amsfonts,bm}

\def\eqref#1{equation~\ref{#1}}
\def\1{\bm{1}}

\DeclareMathAlphabet{\mathsfit}{\encodingdefault}{\sfdefault}{m}{sl}
\SetMathAlphabet{\mathsfit}{bold}{\encodingdefault}{\sfdefault}{bx}{n}

%% file: main.tex
\begin{abstract}
Recent studies have explored autoregressive models for image generation,  with promising results,  and have combined diffusion models with autoregressive frameworks to optimize image generation via diffusion losses. 
In this study,  we present a theoretical analysis of diffusion and autoregressive models with diffusion loss,  highlighting the latter's advantages. 
We present a theoretical comparison of conditional diffusion and autoregressive diffusion with diffusion loss,  demonstrating that patch denoising optimization in autoregressive models effectively mitigates condition errors and leads to a stable condition distribution. Our analysis also reveals that autoregressive condition generation refines the condition,  causing the condition error influence to decay exponentially.
In addition,  we introduce a novel condition refinement approach based on Optimal Transport (OT) theory to address ``condition inconsistency''. We theoretically demonstrate that formulating condition refinement as a Wasserstein Gradient Flow ensures convergence toward the ideal condition distribution,  effectively mitigating condition inconsistency. 
Experiments demonstrate the superiority of our method over diffusion and autoregressive models with diffusion loss methods.
\end{abstract}

\section{Introduction}
Diffusion models have demonstrated remarkable performance in image generation and have been widely adopted across various visual generative tasks \citep{song2020denoising, rombach2022high, Saharia2022Palette}. Recently,  due to the impressive reasoning capabilities exhibited by large language models (LLMs),  autoregressive modeling has garnered significant attention. Consequently,  some studies are exploring autoregressive frameworks for image and video generation,  aiming to integrate them with LLMs to build more powerful multimodal models \citep{latentlm2024}.

Recent advancements in autoregressive image generation have shown performance comparable to diffusion models \citep{llamagen2024, var2024, zhou2025draw}. 
However,  most autoregressive image generation methods rely on Vector Quantized Variational Autoencoders (VQ-VAEs \citep{rombach2022high}) to encode visual content into discrete tokens for next-token prediction modeling. 
\citep{vqfree2024} indicate that VQ-based image generation is sensitive to gradient approximation strategies and suffers from quantization errors,  and propose diffusion loss for autoregressive image generation,  effectively pursuing autoregressive image generation without VQ. 
Nevertheless,  a comparative analysis between Conditional diffusion modeling and autoregressive modeling with diffusion loss remains underexplored.

In this study,  we investigate the differences between autoregressive modeling with diffusion loss and conditional diffusion modeling.
Firstly,  we delve into the theoretical underpinnings of patch denoising optimization in autoregressive models for condition error correction. 
We theoretically prove that,  under standard assumptions of Markov property and Gaussian noise in diffusion modeling,  the iterative patch denoising approach leads to a stable condition distribution. 
Furthermore,  our analysis reveals the crucial behavior of the conditional probability gradient,  showing its attenuation as the condition stabilizes. 
Our theoretical exploration demonstrates that patch denoising in autoregressive modeling effectively mitigates condition errors and consequently contributes to improved conditional generation quality in diffusion modeling.
In addition,  we theoretically demonstrate that the sequence of condition variables generated by an autoregressive process effectively refines the condition,  leading to a reduction in the gradient norm of the conditional probability distribution. 
Specifically,  we demonstrate that the influence of the condition on the outcome,  quantified by the gradient norm,  decays exponentially towards a stationary value as the autoregressive iteration progresses.

Building upon these theoretical insights,  we further analyze the issue of ``condition inconsistency'' in autoregressive condition generation,  demonstrating how extraneous information accumulates and hinders optimal patch generation. To address this,  we introduce a novel condition refinement approach grounded in Optimal Transport (OT) theory. We theoretically prove that formulating condition refinement as a Wasserstein Gradient Flow leads to convergence towards the ideal condition distribution,  effectively mitigating condition inconsistency and ultimately enhancing the quality of patch generation within diffusion models.

In the experiments,  we compare our method against other diffusion and autoregressive models with diffusion loss on ImageNet \citep{imagefolder2024}. Results show the superiority of our method over these competitors. We also analyze the denoising process to demonstrate the effectiveness of our method in condition refinement.
Our main contributions and findings are as follows:
\begin{itemize}[leftmargin=*]
\item We theoretically prove that patch denoising optimization in autoregressive models mitigates condition errors and elucidates the attenuation behavior of the conditional probability gradient as the condition stabilizes.
\item We theoretically establish the efficacy of autoregressive condition refinement,  quantifying the exponential decay of the condition's influence on the outcome as autoregressive iteration progresses to a stationary value.
\item We propose a condition refinement method based on Optimal Transport theory,  and theoretically prove that formulating it as a Wasserstein Gradient Flow ensures convergence towards the ideal condition distribution.
\item Experiments demonstrate our method's superiority over other competitors. Extensive analysis shows the effectiveness of our method in condition refinement.
\end{itemize}

\section{Preliminaries}
\paragraph{Diffusion Modeling.}
Diffusion models are generative frameworks that consist of a forward process. 
The forward (diffusion) process is a Markov chain that transforms data $x_0$ into Gaussian noise $x_T$ through a sequence of Gaussian transitions:
\begin{align}
q(x_{1:T}|x_0) &= \prod_{t=1}^T q(x_t|x_{t-1}), \quad
q(x_t|x_{t-1}) = \mathcal{N}(x_t; \sqrt{1-\beta_t} x_{t-1}, \beta_t \mathbf{I}),
\end{align}
where ${\beta_t}_{t=1}^T$ is a predefined variance schedule with $0 < \beta_1 < \cdots < \beta_T < 1$. 
The reverse (denoising) process reconstructs $x_0$ from $x_T$ via:
\begin{align}
p_\theta(x_{0:T}) &= p(x_T) \prod_{t=1}^T p_\theta(x_{t-1}|x_t), \quad
p_\theta(x_{t-1}|x_t) = \mathcal{N}(x_{t-1}; \mu_\theta(x_t, t), \Sigma_\theta(x_t, t)),
\end{align}
where $\mu_\theta$ and $\Sigma_\theta$ are predicted by a neural network. Since the true posterior $q(x_{t-1}|x_t)$ is intractable, it is approximated using $q(x_{t-1}|x_t, x_0)$ during training.

The model is trained to maximize data likelihood, which is approximated via a variational lower bound. This can be reformulated as a score-matching problem, e.g.,
\begin{align}
\mathbb{E}{x_t \sim p_t(x_t)} \left[ | \nabla{x_t} \log p_t(x_t) - s_\theta(x_t, t) |^2 \right].
\end{align}

\paragraph{Autoregressive Modeling.}
Autoregressive (AR) models are generative frameworks that sequentially predict each element in a data sequence by conditioning on all preceding elements. These models assume that each data point $x_i$ depends only on the prior points $x_{<i} = \{x_1, x_2, \ldots, x_{i-1}\}$. The conditional and joint probabilities can be expressed as:
\begin{align}
    p(x) = p(x_1, x_2, \ldots, x_n) = \prod_{i=1}^n p(x_i | x_{<i})
\end{align}
The generation process starts from $x_1$ and proceeds sequentially to $x_n$, with each step conditioned on all previously generated elements. Related work is in Appendix~\ref{app:related}.

\section{Theoretical Analysis on Autoregressive Image Modeling with Diffusion Loss}
\subsection{Difference of Diffusion Models}
Diffusion models demonstrate exceptional capabilities in generating high-quality visual content. Recently,  autoregressive modeling integrated with diffusion loss has shown significant potential in image generation. We will elucidate the differences between standard conditional diffusion modeling and autoregressive modeling with diffusion loss.

\paragraph{Conditional Diffusion Modeling.}
In traditional conditional diffusion models,  the reverse process is conditioned on a single,  static condition $c$.  
This can be formally expressed as:
\begin{align}
\label{eq:cond_diffusion}
x_{t-1} \sim p(x_{t-1} | x_t,  c), 
\end{align}
where $c$ represents a global condition that influences every step of the denoising trajectory.  
When dealing with images,  we can extend this to an individual patch $x_i$:
\begin{align}
\label{eq:cond_diffusion_element}
x_{i, t-1} \sim p(x_{i, t-1} | x_{j,  t},   c),  \quad \forall j \in \{1,  \dots,  n\}, 
\end{align}
where $n$ represents the number of patches. 
Each patch is denoised based on the same shared condition $c$,  irrespective of its position within the image.

\paragraph{Autoregressive Modeling with Diffusion Loss.}
Autoregressive modeling with diffusion loss allows the condition to evolve autoregressively. 
Instead of a fixed condition $c$,  a sequence of conditions $\{c_i\}$ depends on preceding conditions $\{c_{<i}\}$ including the initial condition $c_0$,  i.e., 
\begin{align}\label{eq:ar_condition}
c_t \sim p(c_t | c_{<i},  c_0), 
\end{align}
where $c_{<i}$ denotes all conditions up to time $i-1$,  i.e.,  $\{c_0,  c_1,  \dots,  c_{i-1}\}$. 
For each patch generation $x_i$,  the reverse process is still a denoising process but is guided by the dynamic condition $c_i$,  i.e., 
\begin{align}
    \label{eq:ar_reverse_element}
    x_{i, t-1}\sim p(x_{i, t-1} | x_{i, t},  c_i), 
\end{align}
where $c_i$ represents the condition for $i$-th patch. 
After generating $x_i$,  it is passed as input to the autoregressive model along with the history of conditions $\{c_{<i+1}\}$,  enabling the prediction of the subsequent condition $c_{i+1}$. 

\subsection{Conditional Denoising Model Error Definition}

\paragraph{Conditional Score Matching as an Upper Bound.}
Score matching is central to training diffusion models,  and its loss is linked to the Wasserstein distance between generated and real data \citep{ScoreWasserstein}. Conditional score matching refines this by incorporating conditioning.  Understanding how conditional score matching relates to standard score matching is key to justifying its use.  This section establishes that the standard score matching loss is upper-bounded by its conditional counterpart. This result supports the use of conditional score matching,  suggesting it might lead to a more controlled training process.

\begin{theorem}[Conditional Score Matching Upper Bound]
\label{thm:upper_bound}
The standard score matching loss is upper-bounded by the conditional score matching loss:
\begin{align}
\label{eq:upper_bound_inequality}
&\mathbb{E}_{\mathbf{x}_t \sim p_t(\mathbf{x}_t) }\left[ \left\| \nabla_{\mathbf{x}_t} \log p_t(\mathbf{x}_t) - s_\theta(\mathbf{x}_t,  t) \right\|^2 \right]
\\
&\leq \mathbb{E}_{c\sim p_c(c), \mathbf{x}_t \sim p_t(\mathbf{x}_t|c)} \left[ \left\| \nabla_{\mathbf{x}_t} \log p_t(\mathbf{x}_t | c) - s_\theta(\mathbf{x}_t,  t) \right\|^2 \right]\notag
\end{align}
See Appendix \ref{app:upper_bound_proof} for the proof,  which uses the law of total probability and Jensen's inequality.
\end{theorem}
The conditional score matching loss serves as an upper bound for the standard score matching loss. Consequently,  minimizing the conditional score matching loss indirectly constrains the standard score matching loss from above.

\paragraph{Error in Conditional Score Matching.}
To analyze the error in conditional score matching,  we build upon the score matching error definition from \citep{Adapting}:
\begin{align}
\epsilon_{\text{score}}^2 := \frac{1}{T} \sum_{t=1}^{T} \mathbb{E}_{X \sim q_t} \left[ \| s_t(X) - s_t^*(X) \|_2^2 \right].
\end{align}
To understand how conditioning affects the error structure,  we need to decompose the conditional and unconditional score matching losses. Expanding these losses into their component terms allows us to identify and analyze the specific contributions of conditioning to the overall error. 

\begin{lemma}[Expansion of Score Matching Loss]
\label{lemma:score_loss_expansion}
Expanding the square term in the score matching loss,  we get:
\begin{align}
\label{eq:score_loss_expanded}
&\mathbb{E}_{x_t\sim p_{t}(x_t)} \left[ \left\| \nabla_{x_t} \log p_{t}(x_t) - s_\theta(x_t,  t) \right\|^{2} \right] \notag\\
&= \mathbb{E}_{x_t \sim p_{t}(x_t)} \Big[ \left\| \nabla_{x_t} \log p_{t}(x_t) \right\|^{2} +  \left\| s_\theta(x_t,  t) \right\|^{2} - 2 \left\langle s_\theta(x_t,  t) ,  \nabla_{x_t} \log p_{t}(x_t) \right\rangle \Big].
\end{align}
Similarly,  for conditional score matching:
\begin{align}
\label{eq:cond_score_loss_expanded}
\!\!&\mathbb{E}_{c, x_t\sim p_c(c) p_{t}(x_t|c)} \left[ \left\| \nabla_{x_t} \log p_{t}(x_t|c) - s_\theta(x_t,  t) \right\|^{2} \right] \notag\\
&= \mathbb{E}_{c, x_t\sim p_c(c) p_{t}(x_t|c)} \Big[ \left\| \nabla_{x_t} \log p_{t}(x_t|c) \right\|^{2} +  \left\| s_\theta(x_t,  t) \right\|^{2} - 2 \left\langle s_\theta(x_t,  t) ,  \nabla_{x_t} \log p_{t}(x_t|c) \right\rangle \Big]\!\!
\end{align}
This expansion separates the loss into terms related to the true score,  the estimated score,  and their interaction,  facilitating a more granular error analysis (Detailed Proof in Appendix~\ref{appendix:score_loss_expansion}). 
\end{lemma}

\begin{definition}[Conditional Error Term $\epsilon_c$]
\label{def:epsilon_c}
To specifically measure the impact of conditioning on the true score's magnitude,  we define the conditional error term $\epsilon_c$ as the change in the expected squared norm of the true score due to conditioning,  relative to the unconditional case:
\begin{align}\label{eq:epsilon_c_def}
\epsilon_c = \frac{1}{T}\sum_{t=1}^{T} \mathbb{E}_{x_t \sim p_t(x_t)}\Big[ \mathbb{E}_{c \sim p_t(c|x_t)} \left[\left\| \nabla_{x_t} \log p_{t}(x_t|c) \right\|^{2}\right] -\left\| \nabla_{x_t} \log p_{t}(x_t) \right\|^{2}\Big]
\end{align}
This term quantifies how much the expected squared norm of the true score changes when we move from unconditional to conditional score matching. A positive $\epsilon_c$ would suggest that conditioning increases the magnitude of the true score,  potentially indicating a more complex or refined score function (Detailed Proof in Appendix~\ref{appendix:epsilon_c_derivation}).
\end{definition}

\begin{definition}[Simplified Conditional Error Term $\overline{\epsilon}_c$]
\label{def:bar_epsilon_c}
For a simpler metric focused purely on the magnitude of the conditional true score,  we define the simplified conditional error term $\overline{\epsilon}_c$:
\begin{align}
\label{eq:bar_epsilon_c_def}
\overline{\epsilon}_c= \frac{1}{T} \sum_{t=1}^{T} \mathbb{E}_{c \sim p_c(c), \mathbf{x}_t \sim p_t(\mathbf{x}_t|c)} \left[\left\| \nabla_{x_t} \log p_{t}(x_t|c) \right\|^{2}\right]
\end{align}
$\overline{\epsilon}_c$ directly measures the expected squared norm of the conditional score. Analyzing $\overline{\epsilon}_c$,  along with $\epsilon_c$,  will help understand the behavior of the true score in conditional settings (Detailed Proof in Appendix~\ref{appendix:epsilon_bar_c}). 
\end{definition}

\subsection{Conditional Control Term Analysis.}
We first investigate the uniqueness of the conditional control term under standard diffusion assumptions and Classifier-Free guidance. As described in Classifier-Free guidance \citep{Classifier-Free, MoreControlforFree},  the conditional reverse diffusion process for sampling is given by:
\begin{align}\label{equ:conSampling}
&p(x_{t-1}|x_t)=\mathcal{N}(x_{t-1}; \mu(x_t) + \sigma_t^2 \nabla_{x_t}\log p(c | x_t),  \sigma^2 I) \notag\\
&x_{t-1} = \mu(x_t) + \sigma_t^2 \nabla_{x_t} \log p(c| x_t) + \sigma_t \epsilon,  \epsilon \sim \mathcal{N}(0,  1)
\end{align}
This shows that conditional control introduces an additional term $\sigma_t^2 \nabla_{x_t} \log p(c| x_t)$ to the mean of the reverse process,  compared to the unconditional diffusion sampling. To understand the impact of this conditional term,  we define $f(c_i) = \|\nabla_{x_t}\log p_t(x_t|c_i)\|^2$. We hypothesize that the difference between the expected value of $f(c_i)$ and the expected value of the unconditional score norm isolates the contribution of this conditional control term. This is formalized in the following lemma:

\begin{lemma}[Uniqueness of Conditional Control Term]
\label{lemma:unique_control_term}
Under standard diffusion assumptions and Classifier-Free guidance,  the difference between the expected squared norm of the conditional score function and the unconditional score function isolates the contribution of the conditional control term.  Let $f(c_i) := \|\nabla_{x_t}\log p_t(x_t|c_i)\|^2$. Then, 
\begin{align}
\mathbb{E}[f(c_i)] - \mathbb{E}\|\nabla_{x_t}\log p_t(x_t)\|^2 = \mathbb{E}\|\sigma_t^2\nabla_{x_t}\log p(c|x_t)\|^2
\end{align}
where the expectation is taken over $x_t \sim p_t(x_t)$ and $c \sim p_c(c)$. This lemma indicates that $f(c_i)$,  through its expected difference with the unconditional score norm,  precisely captures the impact of the conditional guidance term $\sigma_t^2\nabla_{x_t}\log p(c|x_t)$ in the diffusion denoising process. (Proof in Appendix~\ref{app:control_term_uniqueness})
\end{lemma}

\subsection{Condition Refinement through Patch Denoising.}
Building upon the observation that incorporating conditions can amplify errors in diffusion models,  we propose an optimization strategy focused on refining the condition using patch-related corrections. Specifically,  we introduce a mechanism where information from each newly generated patch is propagated to the condition of the subsequent patch through an iterative update process,  $c_{i+1} = \mathcal{T}(c_i)$. This autoregressive approach aims to refine the condition during the denoising process iteratively.

To formalize our approach,  we first establish the foundational assumptions under which our model operates.

\begin{assumption}[Markov Property Assumption]
\label{assumption:markov_property}
The reverse diffusion process adheres to the Markov property,  where each state $x_{t-1}$ is conditionally dependent only on the current state $x_t$.
\end{assumption}

\begin{assumption}[Gaussian Distribution Assumption]
\label{assumption:gaussian_distribution}
The conditional probability distribution $p_t(x_{t-1}|x_t)$ in the reverse diffusion process is assumed to be Gaussian,  expressed as: $p_t(x_{t-1}|x_t)=\mathcal{N}(x_{t-1};\mu(x_t), \sigma_t^2I)$.
\end{assumption}

\begin{assumption}[Small Variance Assumption]
\label{assumption:small_variance}
As the number of time steps $T$ becomes sufficiently large,  the variance $\sigma_t^2$ of the conditional distribution $p_t(x_{t-1}|x_t)$ is assumed to be sufficiently small,  approaching zero as $T$ increases. Furthermore,  for simplicity,  we approximate the variance at each step to be equal and denote it as $\sigma^2$.
\end{assumption}

With the above assumptions,  we model the patch refinement process as an iterative update to the condition:
\begin{align}
c_{i+1}=\mathcal{T}(c_{i}), \quad i \in \mathbb{N}
\end{align}
where $c_{i+1}$ represents the condition at the $(i+1)$-th iteration,  corresponding to the refinement based on the $i$-th patch. $\mathcal{T}$ is a diffusion function that governs the transition from the current condition state $c_i$ to the next state $c_{i+1}$,  encapsulating the information propagation from the generated patch to the subsequent condition. The index $i \in \mathbb{N}$ denotes the iteration step,  analogous to discrete time steps. The condition update process forms a discrete-time Markov chain,  as the future state $c_{i+1}$ depends only on the present state $c_i$: $P(c_{i+1}|c_i, c_{i-1}, \dotsc, c_0)=P(c_{i+1}|c_i)$.

From the Gaussian expansion norm in Equation~\eqref{equ:conSampling},  we observe that the probability distribution of $x_t$ is influenced by $x_t$ itself. Our primary goal is to understand the trajectory of the conditional probability gradient as the condition $c_i$ iteratively refines through the diffusion reverse process. Therefore,  we proceed to analyze how the conditional probability gradient evolves with the iterations of $c_i$ within a standard normal conditional distribution setting.

\begin{proposition}[Condition Refinement via Patch Denoising]
\label{prop:patch_refine_condition}
In the diffusion denoising process,  autoregressively refining the condition through patch-related corrections using the iterative update $c_{i+1}=\mathcal{T}(c_{i})$ leads to improved conditional generation quality.(Detailed Proof in Appendix~\ref{app:RefineViaPatchDenoising})
\end{proposition}

\subsection{Autoregressive Modeling Can Refine Condition}
An autoregressive process is defined as follows:
\begin{align}
c_{i+1} = \sum_{j=0}^i a_jc_j + \varepsilon_{i+1},  \quad i \in \mathbb{N}
\end{align}
\begin{assumption}[Basic Assumptions for Autoregressive Process]\label{ass:arprocessing}
    We provide a set of standard assumptions:
    \begin{enumerate}[itemsep=1pt,  topsep=1pt]
    \item $\sum_{i=0}^{\infty} |a_i| < \infty$ and $\sup_{i\in\mathbb{N}} |a_i| < 1$,  i.e.,  the sequence $\{a_n\}$ is convergent.
    \item $\{\varepsilon_i\}_{i=1}^{\infty}$ are independent and identically distributed,  following $\mathcal{N}(0, \sigma^2)$.
    \item $p_t(x_t|c_i)$ has continuous second-order derivatives with respect to $x_t$.
    \item $\|\nabla^2_{x_t}p_t(x_t|c_i)\| \leq K$ for some constant $K>0$ uniformly holds,  i.e.,  is bounded.
    \item $(\mathcal{X},  \|\cdot\|)$ is a separable complete metric space.
    \end{enumerate}
\end{assumption}
\begin{lemma}[Markov Property~\citep{meyn2012markov}~\citep{Bellet2006}]\label{lemma:MarkovProperty}
Under Assumption~\ref{ass:arprocessing},  by defining the state vector $\mathbf{c}_i = (c_i,  c_{i-1},  \ldots,  c_{i-p+1})^\top$,  the sequence $\{\mathbf{c}_i\}_{i\in\mathbb{N}}$ forms a strong Markov chain. Specifically:
\begin{enumerate}[itemsep=1pt,  topsep=1pt]
    \item The transition probability kernel $P(\mathbf{c}_{i+1} \in \cdot|\mathbf{c}_i)$ on the augmented state space satisfies the Feller property.
    \item There exists a unique invariant probability measure $\pi \in \mathcal{P}(\mathcal{X}^p)$ such that
    \begin{align}
        \pi P = \pi,  ~\text{i.e.}\int_{\mathcal{X}^p} P(A|\mathbf{c})\pi(d\mathbf{c}) = \pi(A), ~~\forall A \in \mathcal{B}(\mathcal{X}^p)
    \end{align}
    \item There exist constants $C > 0$ and $\rho \in (0, 1)$ such that for any initial distribution $\mu \in \mathcal{P}(\mathcal{X}^p)$:
    \begin{align}
        \|\mu P^n - \pi\|_{\text{TV}} \leq C\rho^n\|\mu - \pi\|_{\text{TV}},  \quad \forall n \in \mathbb{N}_0
    \end{align}
    where $\|\cdot\|_{\text{TV}}$ denotes the total variation norm,  and $P^n$ denotes the $n$-step transition probability kernel.
\end{enumerate}
In particular,  for any $n \in \mathbb{N}_0$,  we have:
\begin{align}
    \|\mathcal{L}(\mathbf{c}_n) - \pi\|_{\text{TV}} \leq C\rho^n
\end{align}
where $\mathcal{L}(\mathbf{c}_n)$ denotes the distribution of $\mathbf{c}_n$.
Proof can be found in Appendix~\ref{app:markov_lemma_proof}.
\end{lemma}

\begin{lemma}[Regularity of Conditional Probability~\citep{DU04}]\label{lemma:RegularityofConditionalProbability}
Under Assumption~\ref{ass:arprocessing},  there exist constants $\delta,  L > 0$ such that:
\begin{enumerate}[itemsep=2pt,  topsep=2pt]
    \item $p_t(x_t|c_i) \geq \delta$ for all $(x_t, c_i) \in \mathcal{X} \times \mathcal{X}$.
    \item $\|\nabla_{x_t}p_t(x_t|c_1) - \nabla_{x_t}p_t(x_t|c_2)\| \leq L\|c_1 - c_2\|$ for all $x_t, c_1, c_2 \in \mathcal{X}$.
\end{enumerate}
Proof can be found in Appendix~\ref{app:regularity_proof}.
\end{lemma}
\begin{lemma}[Bounded Derivative Theorem]\label{lemma:BoundedDerivativeTheorem}
On a fixed bounded closed interval $[a, b]$,  if the second derivative is bounded,  then there exist constants $M_1, M_2> 0$ such that:
 \begin{enumerate}[itemsep=2pt,  topsep=2pt]
    \item$\sup_{x_t}\|\nabla_{x_t} p_t(x_t|c)\|<M_1$,  i.e.,  the first derivative is bounded.
    \item$\sup_{x_t}|p_t(x_t|c)|<M_2$,  i.e.,  the original function is bounded.
 \end{enumerate}
Proof can be found in Appendix~\ref{app:bounded_derivative_proof}.
\end{lemma}
\begin{theorem}[Descent of Gradient Norm in Autoregressive Process]
\label{thm:gradient_descent_ar}
Under Assumptions~\ref{ass:arprocessing} and Lemmas~\ref{lemma:MarkovProperty},  \ref{lemma:RegularityofConditionalProbability},  \ref{lemma:BoundedDerivativeTheorem},  there exist constants $M > 0$ and $\beta \in (0, 1)$ such that for any $x_t \in \mathcal{X}$ and $i \in \mathbb{N}_0$:
\begin{align}
\|\nabla_{x_t}\log p_t(x_t|c_i)\| \leq M\beta^i + m
\end{align}
where $m$ is a constant representing the stationary gradient norm (Proof in Appendix~\ref{app:DescentofGradientNorm}).
\end{theorem}

\section{Autoregressive Condition Optimization}
Although autoregressive methods provide contextual information, they inevitably accumulate extraneous noise, leading to ``condition inconsistency''.

\paragraph{Why Optimal Transport?}
We employ Optimal Transport (OT) to rectify this distributional drift for three theoretical reasons:
\begin{enumerate}[leftmargin=*]
    \item \textbf{Geometric Correction:} Unlike overlap-based metrics (e.g., KL divergence), OT quantifies the geometric cost required to transform the noisy generated distribution back to the ideal one.
    \item \textbf{Least Action Principle:} Formulating the refinement as a Wasserstein Gradient Flow identifies the optimal path to eliminate inconsistency while preserving valid semantic information.
    \item \textbf{Convergence:} The framework guarantees theoretical convergence to the stationary ideal distribution, effectively acting as a mathematically grounded ``denoising'' step for the condition.
\end{enumerate}
The full algorithm is provided in Appendix~\ref{app:algorithm}.

\subsection{Condition Inconsistency in Autoregressive Generation}
\begin{wrapfigure}{r}{0.45\textwidth}
    \centering
    \vspace{-4.5mm}
    \includegraphics[width=1\linewidth]{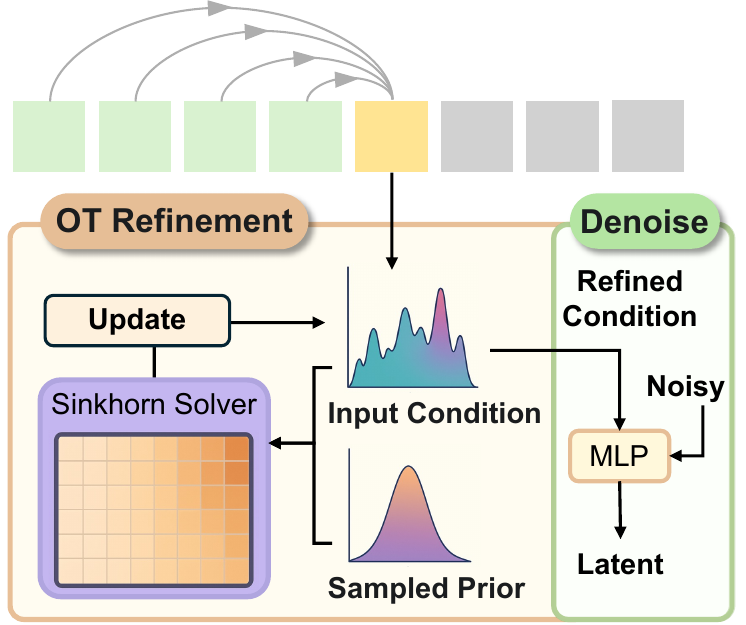}
    \vspace{-6mm}
    \caption{\small The autoregressive model predicts an initial condition, which is processed by the OT Refinement module using a sampled prior derived from Algorithm \ref{alg:aco_denoise_full}. The resulting refined condition then guides the Denoise MLP for latent generation.}
    \label{fig:model}
    \vspace{-4.5mm}
\end{wrapfigure}

The autoregressive condition generation process,  as defined by Equation \eqref{eq:ar_condition},  sequentially constructs conditions,  aiming to capture contextual dependencies. However,  this sequential nature can lead to conditions that are not only influenced by relevant preceding patches but also by accumulated information that is extraneous to generating the current patch. This phenomenon,  which we term ``condition inconsistency'',  arises because the autoregressive process,  while capturing dependencies,  does not inherently guarantee that each generated condition $c_i$ is optimally focused on information strictly necessary for the corresponding patch $x_i$.

\begin{lemma}[Condition Information Inconsistency and Extraneous Information Accumulation]
\label{lemma:condition_inconsistency_information}
Let $c_i = \Phi_\theta(c_{i-1}) + \Gamma_\theta(\epsilon_i)$ be the autoregressively generated condition for patch $x_i$, and $c_i^* = \pi_{\mathcal{I}_i^*}(c_i)$ be its projection onto the minimal sufficient information subspace $\mathcal{I}_i^*$ (derived from $x_{<i}$). The generated $c_i$ inherently contains an extraneous information component $\eta_i = c_i - c_i^*$. This component $\eta_i$ is generally non-zero (i.e., $\mathbb{E}[\|\eta_i\|_2^2] > 0$), accumulating from $(I - \pi_{\mathcal{I}_i^*})\Phi_\theta(c_{i-1})$ and noise components outside $\mathcal{I}_i^*$. The actual conditional distribution $p(x_i|c_i)$ deviates from the ideal $p(x_i|c_i^*)$, and the conditional score $\nabla_{x_t} \log p(x_t|c_i)$ is perturbed from its optimal form under $c_i^*$.
\end{lemma}

\begin{proof}
Let $\mathcal{I}_i^* \subseteq \mathbb{R}^d$ denote the minimal sufficient information subspace for generating patch $x_i$,  and $\pi_{\mathcal{I}_i^*}(\cdot)$ be the orthogonal projection onto this subspace. The ideal condition $c_i^*$ satisfies:
\begin{equation}
    c_i^* = \pi_{\mathcal{I}_i^*}(c_i) \quad \text{where} \quad \mathcal{I}_i^* = \text{span}\{f_k(x_{<i})\}_{k=1}^K
\end{equation}
for some basis functions $\{f_k\}$ encoding relevant dependencies from preceding patches $x_{<i}$. The autoregressive condition generation follows a Markov process:
\begin{equation}
    c_i = \Phi_\theta(c_{i-1}) + \Gamma_\theta(\epsilon_i)
    \label{eq:ar_condition_generation}
\end{equation}
where $\Phi_\theta: \mathbb{R}^d \to \mathbb{R}^d$ is the learned transition operator and $\Gamma_\theta$ modulates the noise injection. The extraneous information component $\eta_i$ can be quantified through subspace decomposition:
\begin{equation}
    \eta_i = (I - \pi_{\mathcal{I}_i^*})c_i = \sum_{k=K+1}^\infty \langle c_i,  v_k \rangle v_k
    \label{eq:extraneous_decomposition}
\end{equation}
where $\{v_k\}$ forms an orthonormal basis for $\mathbb{R}^d$ with the first $K$ vectors spanning $\mathcal{I}_i^*$. The $\ell^2$-norm of extraneous information $\|\eta_i\|_2$ satisfies:
\begin{equation}
    \mathbb{E}[\|\eta_i\|_2^2] = \mathbb{E}\left[\left\|(I - \pi_{\mathcal{I}_i^*})\Phi_\theta(c_{i-1})\right\|_2^2\right] + \text{tr}(\Gamma_\theta\Gamma_\theta^\top)
    \label{eq:extraneous_variance}
\end{equation}

The first term represents propagated extraneous information from previous conditions,  while the second term quantifies newly introduced noise. For the denoising process $\mathcal{D}_t$ at timestep $t$,  the conditional score function becomes perturbed:
\begin{equation}
    \nabla_{x_t} \log p(x_t|c_i) = \underbrace{\nabla_{x_t} \log p(x_t|c_i^*)}_{\text{Ideal score}} + \underbrace{\mathbf{J}_{\eta_i} \nabla_{c_i} \log p(x_t|c_i)}_{\text{Perturbation term}}
    \label{eq:perturbed_score}
\end{equation}
where $\mathbf{J}_{\eta_i}$ is the Jacobian of the perturbation. The extraneous information induces an $\mathcal{O}(\|\eta_i\|_2)$ deviation from the ideal denoising trajectory. The accumulated effect over $N$ patches yields total inconsistency:
\begin{equation}
    \mathcal{E}_{\text{total}} = \sum_{i=1}^N \mathbb{E}\left[\text{OT}_{\lambda}(p(x_i|c_i),  p(x_i|c_i^*))\right]
    \label{eq:total_incostistency}
\end{equation}
where $\text{OT}_{\lambda}$ denotes the Sinkhorn divergence with regularization parameter $\lambda$. This completes the proof of condition information inconsistency.
\end{proof}

\subsection{Optimal Transport for Condition Refinement via Wasserstein Gradient Flow}
\label{subsec:ot_refinement}

Building upon the condition inconsistency analysis in Lemma~\ref{lemma:condition_inconsistency_information},  we present a principled solution through optimal transport theory. Our approach establishes direct connections between the Wasserstein gradient flow framework and condition refinement in autoregressive generation.

\begin{proposition}[Optimal Transport as Wasserstein Gradient Flow]
\label{prop:wgf_refinement}
The condition refinement process can be formulated as a Wasserstein gradient flow that minimizes:
\begin{align}
\mathcal{F}(P_c) := W_2^2(P_c,  P_{c^*}) + \lambda \mathbb{E}_{c\sim P_c}[\|c - \mathcal{T}^{-1}(x)\|^2]
\end{align}
where $P_{c^*}$ denotes the ideal condition distribution and $\mathcal{T}^{-1}$ represents the inverse process of information accumulation in Equation~\eqref{eq:ar_condition_generation}. The solution admits an implementable discrete-time scheme through JKO iterations \citep{JKO}:
\begin{align}
P_c^{(k+1)} = \arg\min_{P} W_2^2(P,  P_c^{(k)}) + \eta_k \mathcal{F}(P)
\end{align}
\end{proposition}

\begin{proof}
Let $\mathcal{P}_2(\mathbb{R}^d)$ denote the space of probability measures with finite second moments. We consider the energy functional:
\begin{align}
\mathcal{F}(P) = \frac{1}{2}W_2^2(P,  P_{c^*}) + \lambda \mathbb{E}_{c\sim P}[\phi(c)]
\end{align}
where $\phi(c) = \|c - \mathcal{T}^{-1}(x)\|^2$ encodes the inverse process regularization. The Wasserstein gradient flow $\partial_t P_t = -\nabla_{W_2}\mathcal{F}(P_t)$ can be discretized via the Jordan-Kinderlehrer-Otto (JKO) scheme:
\begin{align}
P^{(k+1)} = \arg\min_{P} \left\{W_2^2(P,  P^{(k)}) + 2\eta_k\mathcal{F}(P)\right\}
\end{align}
Substituting our specific energy functional yields the update rule in Proposition~\ref{prop:wgf_refinement}. The first term maintains proximity to previous iterates while the second term drives the distribution toward both the ideal condition and inverse-process consistency.

The optimal transport plan between $P^{(k)}$ and $P^{(k+1)}$ corresponds to the McCann interpolant:
\begin{align}
c^{(k+1)} = c^{(k)} - \eta_k\left[\nabla W_2^2(\cdot,  P_{c^*})|_{c^{(k)}} + \lambda\nabla\phi(c^{(k)})\right]
\end{align}
Implementing this requires solving the regularized OT problem:
\begin{align}
\inf_{\gamma\in\Gamma(P^{(k)},  P_{c^*})} \mathbb{E}_{(c, c')}[\|c-c'\|^2] + \epsilon \text{KL}(\gamma|\pi)
\end{align}
where $\pi$ is the independent coupling and $\epsilon$ controls entropy regularization. This leads to the Sinkhorn algorithm implementation described in Proposition~\ref{prop:wgf_refinement}.
\end{proof}

\begin{theorem}[Convergence of Wasserstein Gradient Flow]
\label{thm:main}
Under the assumptions of Proposition~\ref{prop:wgf_refinement},  the Wasserstein gradient flow defined by the energy functional \(\mathcal{F}(P)\) converges to the ideal condition distribution \(P_{c^*}\). Specifically,  for any initial distribution \(P_c^{(0)} \in \mathcal{P}_2(\mathbb{R}^d)\),  the sequence of distributions \(\{P_c^{(k)}\}_{k=1}^\infty\) generated by the JKO scheme satisfies:
\begin{align}
W_2(P_c^{(k)},  P_{c^*}) \leq \rho^k W_2(P_c^{(0)},  P_{c^*}), 
\end{align}
where \(\rho < 1\) is the contraction rate determined by the regularization parameter \(\lambda\) and the step size \(\eta_k\).
\end{theorem}

\begin{proof}[Proof Sketch]
The proof follows from the contractive properties of the Wasserstein gradient flow for convex energy functionals. By the JKO scheme and the regularization term \(\lambda \mathbb{E}_{c\sim P_c}[\|c - \mathcal{T}^{-1}(x)\|^2]\),  the sequence \(\{P_c^{(k)}\}\) forms a Cauchy sequence in the Wasserstein space \(\mathcal{P}_2(\mathbb{R}^d)\). The contraction rate \(\rho\) arises from the strong convexity of the energy functional \(\mathcal{F}(P)\) and Lipschitz continuity of the gradient flow.
\end{proof}

\begin{remark}
The inverse process regularization $\mathcal{T}^{-1}$ directly counters the extraneous information accumulation characterized in Equation~\eqref{eq:extraneous_variance}. By Theorem~\ref{thm:main},  the refinement ensures monotonic improvement in patch generation quality:
\begin{align}
\mathbb{E}[\text{OT}_\lambda(p(x_i|c_i^{(k)}),  p(x_i|c_i^*))] \leq \rho^k \mathbb{E}[\text{OT}_\lambda(p(x_i|c_i^{(0)}),  p(x_i|c_i^*))]
\end{align}
where $\rho < 1$ quantifies the contraction rate of our OT-based refinement operator.

This theorem ensures that the proposed Wasserstein gradient flow refinement process monotonically reduces the Wasserstein distance between the autoregressive condition distribution \(P_c^{(k)}\) and the ideal condition distribution \(P_{c^*}\). The contraction rate \(\rho\) quantifies the refinement efficiency,  with smaller values of \(\rho\) indicating faster convergence.
\end{remark}

\section{Experiments}

\begin{wraptable}{r}{0.42\textwidth}\small
\vspace{-26mm}
\centering
\caption{\small Comparison of different methods on various metrics on ImageNet 256$\times$256 conditional generation. Baseline (CDM) denotes a baseline of conditional diffusion modeling.}
\label{tab:comparison}
\vspace{1mm}
\setlength{\tabcolsep}{3.5pt}
\resizebox{0.42\textwidth}{!}{
\begin{tabular}{lcccc}
\toprule
\bf Method & \bf FID $\downarrow$ & \bf IS $\uparrow$ & \bf Pre. $\uparrow$ & \bf Rec. $\uparrow$ \\
\midrule
LDM-4 \citep{rombach2022high} & 3.60 & 247.7 & 0.87 & 0.48 \\
U-ViT-H/2-G \citep{Bao22All} & 2.29 & 263.9 & 0.81 & 0.62 \\
DiT-XL/2 \citep{DiT-XL2} & 2.27 & 278.2 & 0.83 & 0.57 \\
DiffiT \citep{DiffiT} & 1.73 & 276.5 & 0.80 & 0.62 \\
MDTv2-XL/2 \citep{MDTv2} & 1.58 & 314.7 & 0.79 & 0.65 \\
GIVT \citep{GIVT} & 3.35 & - & 0.84 & 0.53 \\
MAR \citep{vqfree2024}  & 1.55 & 303.7 & 0.81 & 0.62 \\
De-MAR \citep{yao2025denoising} & 1.47 & 305.8 & 0.83 & 0.62 \\
RAR \citep{yu2025randomized} & 1.50 & 306.9 & 0.80 & 0.62 \\
Baseline (CDM) & 3.26 & 259.6 & 0.81 & 0.58 \\
Baseline (AR) & 2.02 & 282.6 & 0.80 & 0.59 \\
MAR \citep{vqfree2024}  & 1.55 & 303.7 & 0.81 & 0.62 \\\midrule
Ours (AR) &  1.52 &  317.6 & 0.82 & 0.60 \\
Ours (MAR) & \bf 1.31 & \bf 324.2 & 0.81 & 0.63 \\
\bottomrule
\end{tabular}}
\vspace{-4mm}
\end{wraptable}

\subsection{Experimental Settings}
Our autoregressive model is directly based on GPT-XL,  while the denoising module is implemented using the MAR-based denoising module.
For the Variational Autoencoder (VAE) component,  we use the KL-16 version of LDM~\citep{rombach2022high}.
Our experiments are conducted on the ImageNet dataset~\citep{ImageNet},  with image resolutions set to $256 \times 256$. 
For evaluation,  we adopt Fréchet Inception Distance (FID)~\citep{FID},  Inception Score (IS)~\citep{IS},  as well as Precision and Recall metrics~\citep{dhariwal2021diffusion}.
During training,  the noise schedule follows a cosine shape and consists of 1000 steps. The learning rate is set to $1\times10^{-5}$,  with a total of 400 epochs and a batch size of 2048. The models are trained with a 100-epoch linear learning rate warmup. We use an exponential moving average (EMA) in parameters with a momentum of 0.9999.

\subsection{Perfomance Comparison}
\begin{wrapfigure}{r}{0.45\textwidth}
\centering
\vspace{-14.5mm}
\includegraphics[width=\linewidth]{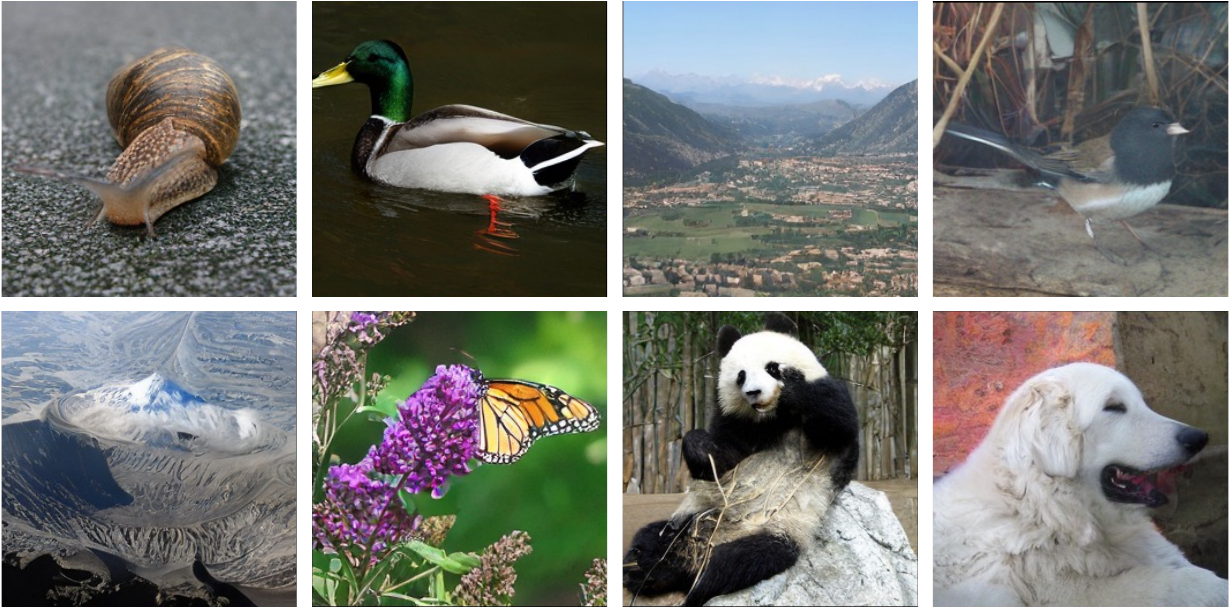}
\vspace{-6.5mm}
\caption{\small Qualitative results on $256 \times 256$ ImageNet class-conditional generation. These images are generated by Ours.}
\label{fig:Qualitative}
\vspace{-3mm}
\end{wrapfigure}

Table~\ref{tab:comparison} shows a comparison of our method against state-of-the-art approaches on ImageNet $256 \times 256$ conditional generation.
Our method achieves the best FID score of \textbf{1.52}, outperforming MAR~\citep{vqfree2024} (1.55), MDTv2-XL/2~\citep{MDTv2} (1.58), and DiffiT~\citep{DiffiT} (1.73). Based on MAR, it can further reach 1.31.
In IS,  our method also achieves the highest score. 
This demonstrates that our model produces samples with higher fidelity and better alignment with the real image distribution. 
For Precision and Recall, it attains 0.81 and 0.63, respectively, remaining competitive with other methods. 
Compared to the baseline,  our model exhibits significant improvements across all evaluation metrics,  highlighting the effectiveness of our approach. Qualitative results are shown in Figure~\ref{fig:Qualitative}.

\begin{wraptable}{r}{0.45\textwidth}\small
    \vspace{-15mm}
    \centering
    \caption{\small Comparison of scalability across different model sizes on ImageNet $256 \times 256$. Our method consistently achieves lower FID and higher IS compared to MAR.}
    \label{tab:scalability_size}
    \vspace{1mm}
    \resizebox{\linewidth}{!}{
    \begin{tabular}{c l c c}
    \toprule
    \textbf{Size} & \textbf{Method} & \textbf{FID} $\downarrow$ & \textbf{IS} $\uparrow$ \\
    \midrule
    \multirow{2}{*}{208M} & MAR \citep{vqfree2024} & 2.31 & 281.7 \\
                          & \textbf{Ours} & \textbf{1.96} & \textbf{290.5} \\
    \midrule
    \multirow{2}{*}{479M} & MAR \citep{vqfree2024} & 1.78 & 296.0 \\
                          & \textbf{Ours} & \textbf{1.59} & \textbf{301.5} \\
    \midrule
    \multirow{2}{*}{943M} & MAR \citep{vqfree2024} & 1.55 & 303.7 \\
                          & \textbf{Ours} & \textbf{1.31} & \textbf{324.2} \\
    \bottomrule
    \end{tabular}}
    \vspace{-2mm}
    \centering
    \caption{Performance comparison on high-resolution ImageNet $512 \times 512$.}
    \label{tab:scalability_res}
    \vspace{1mm}
    \begin{tabular}{l c c}
    \toprule
    \textbf{Method} & \textbf{FID} $\downarrow$ & \textbf{IS} $\uparrow$ \\
    \midrule
    MAR \citep{vqfree2024} & 1.73 & 279.9 \\
    \textbf{Ours} & \textbf{1.58} & \textbf{302.3} \\
    \bottomrule
    \end{tabular}
    \vspace{-8mm}
\end{wraptable}

\subsection{Scalability Analysis}
\label{sec:scalability}
To investigate the scalability and robustness of our proposed method, we conducted additional evaluations across varying model sizes and higher image resolutions. These experiments aim to verify whether the benefits of our Condition Refinement approach persist as the model capacity increases and the generation task becomes more challenging.

\vspace{-5mm}
\paragraph{Scalability across Model Sizes.} 
We evaluated our method against the strong baseline MAR \citep{vqfree2024} on ImageNet $256 \times 256$ using three different model scales: 208M, 479M, and 943M parameters. As presented in Table~\ref{tab:scalability_size}, our method consistently outperforms MAR across all model sizes. Notably, the performance gap widens as the model size increases, suggesting that our autoregressive condition optimization effectively leverages larger capacities for superior generation quality.

\vspace{-5mm}
\paragraph{High-Resolution Generation.}
To further assess the generalization capability of our method, we extended our evaluation to a higher resolution setting on ImageNet $512 \times 512$ (using a model size of approximately 481M parameters). Table~\ref{tab:scalability_res} demonstrates the superiority of our approach in high-resolution synthesis. Our method achieves an FID of 1.58 compared to 1.73 for MAR, indicating that our OT-based condition refinement remains effective in mitigating inconsistencies even in higher-dimensional spaces.

\subsection{Condition Errors Analysis}
\begin{wrapfigure}{r}{0.5\textwidth}
\centering
\vspace{-14mm}
\includegraphics[width=\linewidth]{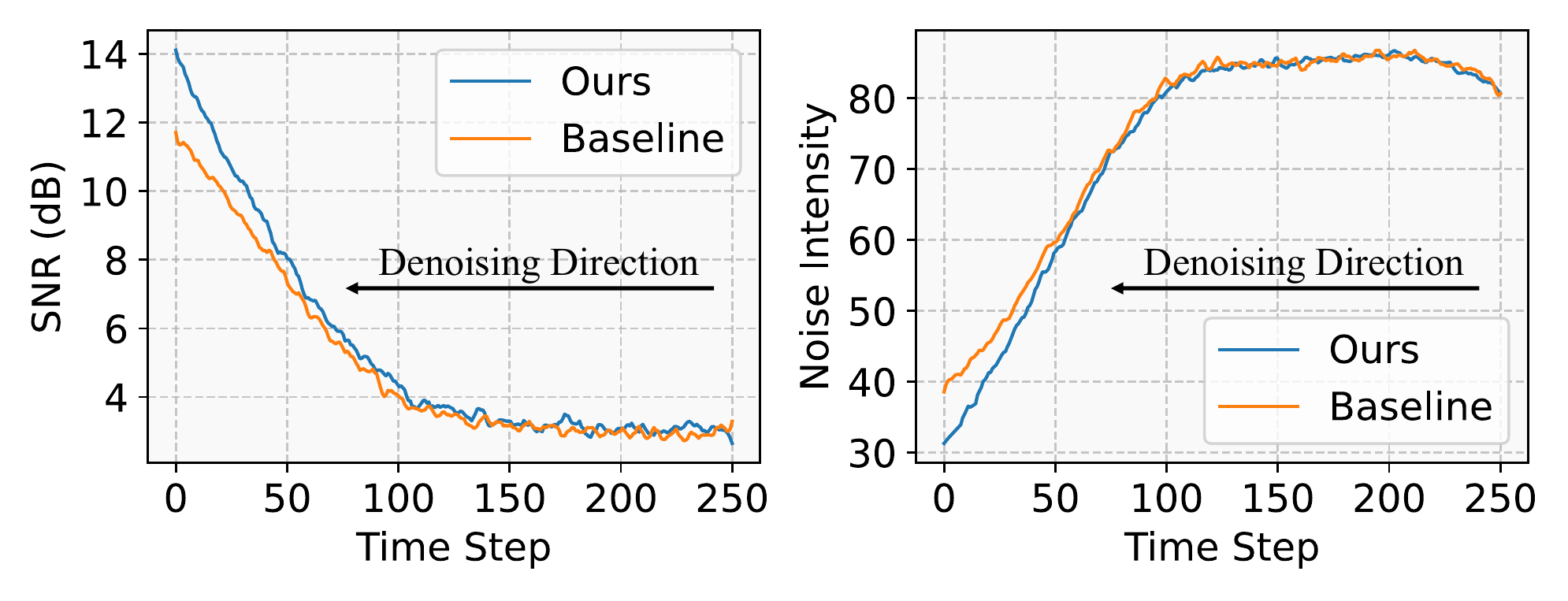}
\vspace{-7mm}
\caption{\small Analysis of Signal-to-Noise Ratio (SNR,  {\bf Left}) and Noise Intensity ({\bf Right}) during the denoising process of our method and the baseline. All analyses are computed in the image space after VAE decoding.}
\label{fig:noise}
\vspace{-2mm}
\end{wrapfigure}

Figure~\ref{fig:noise} presents the denoising analysis, showing Signal-to-Noise Ratio (SNR) and Noise Intensity over time for both our method and a baseline. 
The denoising process proceeds from right to left, with time steps decreasing as denoising progresses.
Our method consistently achieves higher SNR, with a widening gap in later stages. 
Similarly,  the right panel shows the Noise Intensity,  where both methods show a reduction in noise as denoising progresses.  
Consistent with the SNR analysis,  our method exhibits a marginally lower Noise Intensity,  especially in the earlier time steps, i.e., later stages of denoising. 
These results highlight the efficacy of our proposed OT-based refinement of conditional distributions and effectively mitigate the potential inconsistencies introduced by purely autoregressive methods.  

\section{Conclusion}
In this work, our analysis shows that patch denoising in autoregressive models mitigates condition errors, stabilizing condition distribution and enhancing generation quality. Autoregressive condition generation further refines conditions, exponentially reducing error influence. 
To address ``condition inconsistency'', we introduce a refinement method based on Optimal Transport and prove that casting it as Wasserstein Gradient Flow ensures convergence. 
Experimental results and analysis show the effectiveness of our method.

%% file: appendix.tex
\clearpage
\appendix
\section{Related Work}\label{app:related}
\subsection{Diffusion Model}
Diffusion models have emerged as a powerful generative framework,  surpassing GANs~\citep{GAN} and VAEs~\citep{VAE} in stability and sample quality. DDPMs~\citep{ho2020denoising} introduced a noise-based training and reconstruction paradigm,  later linked theoretically to Score Matching and DAEs~\citep{ScoreMatching}. However,  early diffusion models suffered from slow sampling due to numerous iterative steps. Improved DDPMs~\citep{nichol2021improved} refined noise scheduling,  while DDIMs~\citep{song2020denoising} accelerated generation through a non-Markovian formulation. LDMs~\citep{rombach2022high} further optimized efficiency by applying diffusion in a lower-dimensional latent space. Diffusion models also exhibit theoretical advantages over GANs,  notably their implicit minimization of the Wasserstein distance~\citep{ScoreWasserstein},  leading to better convergence and robustness. Enhancing conditional control remains a key research focus: Classifier-Free Diffusion Guidance~\citep{ho2022classifier} enables flexible conditioning without external classifiers, and structure-aware adaptations~\citep{Adapting} improve efficiency for structured data. Various applications extend their utility: Palette~\citep{Saharia2022Palette} enhances image restoration,  GLIDE~\citep{Nichol2022GLIDE} improves text-guided synthesis,  CDMs~\citep{Ho2022Cascaded} refine images progressively,  and ControlNet~\citep{Zhang2023ControlNet} integrates structural conditions for enhanced controllability. Detecting diffusion-generated images is increasingly challenging,  with studies like Schaefer et al.~\citep{Schaefer2023Detection} highlighting the need for robust detection methods.  

\subsection{Autoregressive Image Generation}
Autoregressive models,  despite their effectiveness,  face computational constraints due to sequential generation \citep{zhou2024less,sun2025speed,zhou2025medical}. Optimization efforts focus on efficiency and scalability:  LlamaGen~\citep{llamagen2024} leverages large-scale training to surpass diffusion models in quality and efficiency,  while VAR~\citep{var2024} reduces inference latency via next-scale prediction. Spatial alignment strategies like ImageFolder~\citep{imagefolder2024} improve autoregressive modeling,  and Emu3~\citep{emu32024} unifies token prediction across modalities. Expanding autoregression to multimodal tasks requires bridging discrete and continuous data representations. Lumina-mGPT~\citep{luminamgpt2024} employs a decoder-only Transformer for high-quality text-to-image synthesis,  while MMAR~\citep{mmar2024} models continuous tokens to enhance understanding and generation. Traditional vector quantization in autoregressive models is being reconsidered:  VQ-free autoregression~\citep{vqfree2024} introduces diffusion-based per-token probabilities for efficiency,  and LatentLM~\citep{latentlm2024} integrates next-token diffusion for multimodal synthesis across image,  speech,  and text.

\section{Limitations}\label{app:limitation}

While our research provides novel theoretical insights and algorithmic advancements, it is important to acknowledge certain limitations. Specifically, our experimental evaluation has not been conducted on large-scale models due to the substantial computational resources required for such validation. Instead, our focus has been on a rigorous theoretical analysis and the development of scalable algorithms. Despite the absence of experiments on very large models, the generality and applicability of our method have been theoretically established, and our experiments on general settings support the soundness of the proposed approach. We believe future work can extend these evaluations to more resource-intensive settings to further verify empirical performance at scale.

\section{Proof of Theorem \ref{thm:upper_bound} (Conditional Score Matching Upper Bound)}
\label{app:upper_bound_proof}

We begin by establishing two foundational lemmas required for the proof.

\begin{lemma}[Bayes' Theorem for Conditional Scores]
\label{lem:bayes}
For any measurable sets $c$ and $\mathbf{x}_t$,  the posterior distribution satisfies:
\begin{align*}
p_t(c|\mathbf{x}_t) = \frac{p_c(c)p_t(\mathbf{x}_t|c)}{p_t(\mathbf{x}_t)} = \frac{p_c(c)p_t(\mathbf{x}_t|c)}{\int p_c(c')p_t(\mathbf{x}_t|c')dc'}
\end{align*}
where the second equality explicitly shows the marginalization over $c'$.
\end{lemma}

\begin{lemma}[Jensen's Inequality for Convex Functions]
\label{lem:jensen}
For any convex function $f:\mathbb{R}^d \to \mathbb{R}$ and random variable $\mathbf{Y}$ with finite expectation:
\begin{align*}
f\left(\mathbb{E}[\mathbf{Y}]\right) \leq \mathbb{E}\left[f(\mathbf{Y})\right]
\end{align*}
Equality holds if and only if $f$ is affine linear on the support of $\mathbf{Y}$,  or $\mathbf{Y}$ is constant almost surely.
\end{lemma}

\paragraph{Step-by-Step Proof:}
Using these lemmas,  we proceed with the main proof.

\paragraph{Step 1: Marginal-Conditional Decomposition.}
Express the marginal distribution through conditioning variables:
\begin{align*}
p_t(\mathbf{x}_t) &= \int_{\mathcal{C}} p_c(c)p_t(\mathbf{x}_t|c)dc \quad \text{(Law of total probability)}
\end{align*}
Differentiate both sides under the integral sign (valid under Dominated Convergence Theorem conditions):
\begin{align*}
\nabla_{\mathbf{x}_t}p_t(\mathbf{x}_t) &= \int_{\mathcal{C}} p_c(c)\nabla_{\mathbf{x}_t}p_t(\mathbf{x}_t|c)dc
\end{align*}

\paragraph{Step 2: Score Function Representation.}
Using Lemma \ref{lem:bayes},  decompose the marginal score:
\begin{align*}
\nabla_{\mathbf{x}_t}\log p_t(\mathbf{x}_t) &= \frac{\nabla_{\mathbf{x}_t}p_t(\mathbf{x}_t)}{p_t(\mathbf{x}_t)} \\
&= \frac{\int p_c(c)\nabla_{\mathbf{x}_t}p_t(\mathbf{x}_t|c)dc}{p_t(\mathbf{x}_t)} \\
&= \int \underbrace{\frac{p_c(c)p_t(\mathbf{x}_t|c)}{p_t(\mathbf{x}_t)}}_{p_t(c|\mathbf{x}_t)} \nabla_{\mathbf{x}_t}\log p_t(\mathbf{x}_t|c)dc \\
&= \mathbb{E}_{c \sim p_t(c|\mathbf{x}_t)}\left[ \nabla_{\mathbf{x}_t}\log p_t(\mathbf{x}_t|c) \right]
\end{align*}
where the critical step (line 3) applies Lemma \ref{lem:bayes} to identify the posterior distribution.

\paragraph{Step 3: Jensen's Inequality Application.}
Substitute into the unconditional loss:
\begin{align*}
\mathbb{E}_{\mathbf{x}_t} \left\| \nabla\log p_t - s_\theta \right\|^2 
&= \mathbb{E}_{\mathbf{x}_t} \left\| \mathbb{E}_{c|\mathbf{x}_t}[\nabla\log p_t(\cdot|c)] - s_\theta \right\|^2 \\
&\leq \mathbb{E}_{\mathbf{x}_t} \mathbb{E}_{c|\mathbf{x}_t} \left\| \nabla\log p_t(\cdot|c) - s_\theta \right\|^2 \quad \text{(by Lemma \ref{lem:jensen})}
\end{align*}
Here we specifically apply Lemma \ref{lem:jensen} with:
\begin{itemize}
\item $f(\mathbf{y}) = \|\mathbf{y} - s_\theta\|^2$ (convex since $\|\cdot\|^2$ is convex)
\item $\mathbf{Y} = \nabla_{\mathbf{x}_t}\log p_t(\mathbf{x}_t|c)$
\end{itemize}

\paragraph{Step 4: Law of Total Expectation.}
Convert the nested expectation to a joint expectation:
\begin{align*}
\mathbb{E}_{\mathbf{x}_t} \mathbb{E}_{c|\mathbf{x}_t}[\cdot] = \mathbb{E}_{c, \mathbf{x}_t}[\cdot] = \mathbb{E}_{c \sim p_c} \mathbb{E}_{\mathbf{x}_t \sim p_t(\cdot|c)}[\cdot]
\end{align*}
Thus, we obtain the final inequality:
\begin{align*}
\mathbb{E}_{\mathbf{x}_t} \left\| \nabla\log p_t - s_\theta \right\|^2 
\leq \mathbb{E}_{c, \mathbf{x}_t} \left\| \nabla\log p_t(\cdot|c) - s_\theta \right\|^2
\end{align*}

\paragraph{Tightness Analysis.} By Lemma \ref{lem:jensen},  equality holds iff $\nabla_{\mathbf{x}_t}\log p_t(\mathbf{x}_t|c)$ is constant $c$-a.s.,  which requires $p_t(\mathbf{x}_t|c) = p_t(\mathbf{x}_t)$ for all $c$ in the support of $p_c$. This corresponds to statistical independence $\mathbf{x}_t \perp\perp c$.

\section{Detailed Derivations for Conditional Score Matching Analysis}\label{appendix:conditional_score_matching_derivations}

\subsection{Proof of Lemma \ref{lemma:score_loss_expansion}}
\label{appendix:score_loss_expansion}
Following the score matching framework from \citep{ScoreMatching},  we begin by expanding the squared norm in both unconditional and conditional score matching objectives.

\paragraph{Unconditional Case:}
For the unconditional score matching loss:
\begin{align*}
&\mathbb{E}_{x_t\sim p_{t}(x_t)} \left[ \left\| \nabla_{x_t} \log p_{t}(x_t) - s_\theta(x_t,  t) \right\|^{2} \right] \\
&= \mathbb{E}_{x_t \sim p_{t}(x_t)} \left[ \left(\nabla_{x_t} \log p_{t}(x_t) - s_\theta(x_t,  t)\right)^\top \left(\nabla_{x_t} \log p_{t}(x_t) - s_\theta(x_t,  t)\right) \right] \\
&= \mathbb{E}_{x_t} \Big[ \underbrace{\left\| \nabla_{x_t} \log p_{t}(x_t) \right\|^{2}}_{\text{True score norm}} + \underbrace{\left\| s_\theta(x_t,  t) \right\|^{2}}_{\text{Learned score norm}} \\
&\quad - 2\underbrace{\left\langle s_\theta(x_t,  t),  \nabla_{x_t} \log p_{t}(x_t) \right\rangle}_{\text{Alignment term}} \Big]
\end{align*}
This follows directly from the identity $\|a-b\|^2 = \|a\|^2 + \|b\|^2 - 2a^\top b$.

\paragraph{Conditional Case:}
For conditional score matching,  expanding the L2 norm of the training error,  we have:
\begin{align*}
&E_{c, x_t\sim p_c(c) p_{t}(x_t|c)} \left[ \left\| \nabla_{x_t} \log p_{t}(x_t|c) - s_\theta(x_t,  t) \right\|^{2} \right] \\
&= E_{c, x_t\sim p_c(c) p_{t}(x_t|c)} \left[ \left\| \nabla_{x_t} \log p_{t}(x_t|c) \right\|^{2} +  \left\| s_\theta(x_t,  t) \right\|^{2} - 2 \left\langle s_\theta(x_t,  t) ,  \nabla_{x_t} \log p_{t}(x_t|c) \right\rangle \right] \\
&= E_{x_t \sim p_t(x_t)} E_{c \sim p_t(c|x_t)} \left[ \left\| \nabla_{x_t} \log p_{t}(x_t|c) \right\|^{2} +  \left\| s_\theta(x_t,  t) \right\|^{2} - 2 \left\langle s_\theta(x_t,  t) ,  \nabla_{x_t} \log p_{t}(x_t|c) \right\rangle \right] \\
&= E_{x_t \sim p_t(x_t)} \left[ E_{c \sim p_t(c|x_t)} \left[\left\| \nabla_{x_t} \log p_{t}(x_t|c) \right\|^{2} \right]+  \left\| s_{\theta}(x_t,  t) \right\|^{2} - 2E_{c \sim p_t(c|x_t)} \left\langle s_\theta(x_t,  t) ,  \nabla_{x_t} \log p_{t}(x_t|c) \right\rangle \right]
\end{align*}

\subsection{Derivation of Error Difference $\epsilon_c$ (Definition~\ref{def:epsilon_c})}
\label{appendix:epsilon_c_derivation}

\paragraph{Loss Difference Analysis:}
Subtracting the unconditional loss from the conditional loss (after expansion):
\begin{align*}
&\mathbb{E}_{c, x_t}[\text{Loss}] - \mathbb{E}_{x_t}[\text{Loss}] \\
&= \mathbb{E}_{x_t} \mathbb{E}_{c|x_t} \left[\left\| \nabla_{x_t} \log p_{t}(x_t|c) \right\|^{2}\right] - \mathbb{E}_{x_t} \left\| \nabla_{x_t} \log p_{t}(x_t) \right\|^{2} \\
&\quad + \underbrace{\mathbb{E}_{x_t} \left\| s_{\theta} \right\|^{2} - \mathbb{E}_{x_t} \left\| s_{\theta} \right\|^{2}}_{=0} \\
&\quad - 2\underbrace{\left( \mathbb{E}_{x_t, c|x_t} \langle s_{\theta},  \nabla\log p_t(x_t|c) \rangle - \mathbb{E}_{x_t} \langle s_{\theta},  \nabla\log p_t(x_t) \rangle \right)}_{\text{Vanishes by tower property}}
\end{align*}

The cross-terms cancel due to the tower property of expectation:
\begin{align*}
\mathbb{E}_{x_t} \mathbb{E}_{c|x_t} \langle s_{\theta},  \nabla\log p_t(x_t|c) \rangle &= \mathbb{E}_{x_t} \langle s_{\theta},  \mathbb{E}_{c|x_t} \nabla\log p_t(x_t|c) \rangle \\
&= \mathbb{E}_{x_t} \langle s_{\theta},  \nabla\log p_t(x_t) \rangle
\end{align*}
where we used the identity $\nabla\log p_t(x_t) = \mathbb{E}_{c|x_t} \nabla\log p_t(x_t|c)$ from \citep{Adapting}.

\paragraph{Final Error Expression:}
Thus,  the difference reduces to:
\begin{align*}
\epsilon_c &= \frac{1}{T} \sum_{t=1}^T \mathbb{E}_{x_t} \left[ \mathbb{E}_{c|x_t} \left\| \nabla\log p_t(x_t|c) \right\|^2 - \left\| \nabla\log p_t(x_t) \right\|^2 \right]
\end{align*}
This quantifies the excess ``score energy'' induced by conditioning,  similar to variance decomposition in probability theory.

\subsection{Properties of $\overline{\epsilon}_c$ (Definition~\ref{def:bar_epsilon_c})}
\label{appendix:epsilon_bar_c}

From Definition \ref{def:bar_epsilon_c},  we can relate $\overline{\epsilon}_c$ to $\epsilon_c$ using the law of total variance:
\begin{align*}
\overline{\epsilon}_c &= \epsilon_c + \frac{1}{T} \sum_{t=1}^T \mathbb{E}_{x_t} \left\| \nabla\log p_t(x_t) \right\|^2 
\end{align*}

This decomposition reveals that $\overline{\epsilon}_c$ contains both the intrinsic score energy from the unconditional distribution and the additional energy $\epsilon_c$ from conditioning.

\section{Conditional Control Term Uniqueness Proof}
\label{app:control_term_uniqueness}

\begin{proof}
We prove Lemma~\ref{lemma:unique_control_term} through three key steps:

\paragraph{Step 1: Bayesian Score Decomposition}  
Using Bayes' rule $p_t(x_t|c) = \frac{p(c|x_t)p_t(x_t)}{p(c)}$,  we derive:
\begin{align*}
\nabla_{x_t}\log p_t(x_t|c) &= \nabla_{x_t}\log p(c|x_t) + \nabla_{x_t}\log p_t(x_t) - \underbrace{\nabla_{x_t}\log p(c)}_{=0} \nonumber \\
&= \nabla_{x_t}\log p(c|x_t) + \nabla_{x_t}\log p_t(x_t)
\end{align*}

\paragraph{Step 2: Cross-Term Cancellation}  
The squared norm decomposes as:
\begin{align*}
\|\nabla_{x_t}\log p_t(x_t|c)\|^2 &= \|\nabla_{x_t}\log p(c|x_t)\|^2 + \|\nabla_{x_t}\log p_t(x_t)\|^2 \nonumber \\
&\quad + 2\langle\nabla_{x_t}\log p(c|x_t),  \nabla_{x_t}\log p_t(x_t)\rangle
\end{align*}

Taking expectation over $p_t(x_t)$ and $p_c(c)$:
\begin{align*}
\mathbb{E}[\langle\nabla_{x_t}\log p(c|x_t),  \nabla_{x_t}\log p_t(x_t)\rangle] &= \mathbb{E}_{x_t}\left[\mathbb{E}_c\left[\langle\nabla_{x_t}\log p(c|x_t),  \nabla_{x_t}\log p_t(x_t)\rangle\mid x_t\right]\right] \nonumber \\
&= \mathbb{E}_{x_t}\left[\langle\underbrace{\mathbb{E}_c[\nabla_{x_t}\log p(c|x_t)]}_{=0},  \nabla_{x_t}\log p_t(x_t)\rangle\right] = 0
\end{align*}
where the inner expectation vanishes because $\mathbb{E}_{c}[\nabla_{x_t}\log p(c|x_t)] = \nabla_{x_t}\mathbb{E}_{c}[p(c|x_t)] = \nabla_{x_t}1 = 0$.

\paragraph{Step 3: Variance Propagation}  
Combining results from Steps 1-2:
\begin{align*}
\epsilon_c &= \mathbb{E}\|\nabla_{x_t}\log p_t(x_t|c)\|^2 - \mathbb{E}\|\nabla_{x_t}\log p_t(x_t)\|^2 \nonumber \\
&= \mathbb{E}\|\sigma_t^2\nabla_{x_t}\log p(c|x_t)\|^2
\end{align*}

The scaling factor $\sigma_t^2$ emerges from the reverse process parameterization in \citep{Classifier-Free, MoreControlforFree},  where the conditional mean adjustment contains an explicit $\sigma_t^2$ multiplier. This completes the proof that $\epsilon_c$ isolates the conditional control term's contribution.
\end{proof}

\section{Proof of Condition Refinement via Patch Denoising
(Proposition~\ref{prop:patch_refine_condition})}\label{app:RefineViaPatchDenoising}
To understand the long-term behavior of the condition refinement process,  we invoke the Markov Chain Stationary Theorem.
The conditional probability density function $p(x|c)$ is then given by:
aBy invoking ~\citep{meyn2012markov},  we establish that as $c_i$ iterates,  its distribution converges to a stationary distribution. 
(The details of the lemma we used can be found in Appendix~\ref{app:markov_lemma}.)
Consequently,  our analysis shifts to understanding how the gradient of the conditional probability evolves when $c_i$ follows a normal distribution. To compute the gradient of the conditional probability density function $p(x|c)$ with respect to $x$,  the joint distribution of $x$ and $c$ follows a multivariate normal distribution:
\begin{align*}
\begin{pmatrix}
X \\
C
\end{pmatrix} \sim \mathcal{N} \left( 
\begin{pmatrix}
\mu_x \\
\mu_c
\end{pmatrix},  
\begin{pmatrix}
\sigma_{xx} & \sigma_{xc} \\
\sigma_{xc} & \sigma_{cc}
\end{pmatrix} \right).
\end{align*}
The conditional probability density function $p(x|c)$ is:
\begin{equation*}
p(x|c) = \frac{1}{\sqrt{2 \pi(\sigma_{xx} -\frac{\sigma_{xc}^2}{\sigma_{cc}}}) } \exp\left( -\frac{(x -(\mu_x + \frac{\sigma_{xc}}{\sigma_{cc}} (c - \mu_c)))^2}{2 (\sigma_{xx}-\frac{\sigma_{xc}^2}{\sigma_{cc}})} \right)
\end{equation*}
Taking the logarithm,  we obtain the log-likelihood function:
\begin{align*}
\log p(x|c) = &-\frac{1}{2} \log \left(2 \pi (\sigma_{xx} - \frac{\sigma_{xc}^2}{\sigma_{cc}}) \right) \notag\\
&- \frac{(x - (\mu_x + \frac{\sigma_{xc}}{\sigma_{cc}} (c - \mu_c)))^2}{2 (\sigma_{xx}-\frac{\sigma_{xc}^2}{\sigma_{cc}})}
\end{align*}
Differentiating with respect to $x$,  while ignoring constant terms,  yields:
\begin{align*}
\nabla_x\log p(x|c) = -\frac{x - (\mu_x + \frac{\sigma_{xc}}{\sigma_{cc}}(c - \mu_c))}{\sigma_{xx} - \frac{\sigma_{xc}^2}{\sigma_{cc}}}.
\end{align*}
The squared norm of this gradient is given by:
\begin{align*}
\|\nabla_x\log p(x|c)\|^2 = \left(\frac{x - (\mu_x + \frac{\sigma_{xc}}{\sigma_{cc}}(c - \mu_c))}{\sigma_{xx} - \frac{\sigma_{xc}^2}{\sigma_{cc}}}\right)^2.
\end{align*}
The conditional mean and the conditional variance primarily influence the gradient behavior. The conditional mean is:
\begin{equation}
\mu_x + \frac{\sigma_{xc}}{\sigma_{cc}} (c - \mu_c), \nonumber
\end{equation}
where $\mu_c$ represents the mean of $c$. As $c$ iterates and reaches its stationary distribution,  $\mu_c$ converges to a constant,  which we denote as $\mu_c^{\text{stable}}$. Consequently,  the conditional mean stabilizes to a fixed value.
The conditional variance $\sigma_{xx} - \frac{\sigma_{xc}^2}{\sigma_{cc}}$ is determined by the covariance structure of the joint distribution. Since this variance does not depend on $c$ in its stationary distribution,  it remains unchanged.
As $c$ reaches its stationary distribution,  the deviation of $x$ from the conditional mean gradually diminishes while the variance remains constant. It leads to gradient magnitude decay,  indicating attenuation of the conditional probability gradient over iterations.

\section{Proof of Descent of Gradient Norm in Autoregressive
Process (Theorem~\ref{thm:gradient_descent_ar})}\label{app:DescentofGradientNorm}
\begin{proof}
\begin{remark}[State Space Representation]\label{rem:state_space}
By introducing the state vector $\mathbf{c}_i = (c_i,  c_{i-1},  \ldots,  c_{i-p+1})^\top$,  we can represent the original process as a vector-valued first-order autoregressive process:
\begin{align*}
\mathbf{c}_{i+1} = \mathbf{A}\mathbf{c}_i + \mathbf{e}_{i+1}
\end{align*}
\begin{align*}
\mathbf{A} = \begin{pmatrix}
a_0 & a_1 & \cdots & a_{p-2} & a_{p-1} \\
1 & 0 & \cdots & 0 & 0 \\
0 & 1 & \cdots & 0 & 0 \\
\vdots & \vdots & \ddots & \vdots & \vdots \\
0 & 0 & \cdots & 1 & 0
\end{pmatrix}, ~
\mathbf{e}_{i+1} = \begin{pmatrix}
\varepsilon_{i+1} \\
0 \\
0 \\
\vdots \\
0
\end{pmatrix}
\end{align*}
In this representation,  the conditional distribution of $\mathbf{c}_{i+1}$ depends only on $\mathbf{c}_i$. From Assumption~\ref{ass:arprocessing} (1),  the spectral radius of matrix $\mathbf{A}$ is less than 1,  which ensures the stability of the process.
\end{remark}
Let $p_t(x_t|\pi) = \int_{\mathcal{X}} p_t(x_t|c)\pi(dc)$.
The log-likelihood gradient can be decomposed as:
\begin{align*}
\nabla_{x_t}\log p_t(x_t|c_i) &= \frac{\nabla_{x_t}p_t(x_t|c_i)}{p_t(x_t|c_i)} \\
&= \left( \frac{\nabla_{x_t}p_t(x_t|c_i)}{p_t(x_t|c_i)} - \frac{\nabla_{x_t}p_t(x_t|\pi)}{p_t(x_t|\pi)} \right) + \underbrace{\frac{\nabla_{x_t}p_t(x_t|\pi)}{p_t(x_t|\pi)}}_{\text{Stationary Term}}
\end{align*}
Using the triangle inequality for norms,  we get:
\begin{align*}
\|\nabla_{x_t}\log p_t(x_t|c_i)\|
&\leq \left\|\frac{\nabla_{x_t}p_t(x_t|c_i)}{p_t(x_t|c_i)} - \frac{\nabla_{x_t}p_t(x_t|\pi)}{p_t(x_t|\pi)}\right\| + \left\| \underbrace{\frac{\nabla_{x_t}p_t(x_t|\pi)}{p_t(x_t|\pi)}}_{\text{Stationary Term}}\right\|\\
\end{align*}
For the non-stationary term,  we have:
\begin{align}
&~~~~~\left\|\frac{\nabla_{x_t}p_t(x_t|c_i)}{p_t(x_t|c_i)} - \frac{\nabla_{x_t}p_t(x_t|\pi)}{p_t(x_t|\pi)}\right\| \nonumber\\
&= \left\|\frac{p_t(x_t|\pi)\nabla_{x_t}p_t(x_t|c_i) - p_t(x_t|c_i)\nabla_{x_t}p_t(x_t|\pi)}{p_t(x_t|c_i)p_t(x_t|\pi)}\right\| \nonumber\\
&\leq \frac{1}{\delta^2}\left\|p_t(x_t|\pi)\nabla_{x_t}p_t(x_t|c_i) - p_t(x_t|c_i)\nabla_{x_t}p_t(x_t|\pi)\right\| \nonumber\\
&= \frac{1}{\delta^2}\left\|p_t(x_t|\pi)[\nabla_{x_t}p_t(x_t|c_i)-\nabla_{x_t}{p_t(x_t|\pi)}] + \nabla_{x_t}p_t(x_t|\pi)[p_t(x_t|\pi)-p_t(x_t|c_i)]\right\| \nonumber\\
& \leq \frac{1}{\delta^2}\Big[\left\|p_t(x_t|\pi)[\nabla_{x_t}p_t(x_t|c_i)-\nabla_{x_t}{p_t(x_t|\pi)}]\| + \|\nabla_{x_t}p_t(x_t|\pi)[p_t(x_t|\pi)-p_t(x_t|c_i)]\right\| \Big]\nonumber\\
\intertext{Using the Lipschitz property of both the conditional probability density function $p_t(x_t|c)$ and its gradient $\nabla_{x_t}p_t(x_t|c)$ with respect to $c$ (with Lipschitz constant $L$),  and the upper bound $M_2$ for $p_t(x_t|\pi)$ from Lemma~\ref{lemma:RegularityofConditionalProbability},  the above inequality becomes:}
& \leq \frac{1}{\delta^2}\Big[ M_{2}L\|c_{i}-c_{\pi}\| + \|\nabla_{x_t}p_t(x_t|\pi)\|L\|c_{i}-c_{\pi}\| \Big]\nonumber\\
&=\frac{L}{\delta^2}\Big[ M_{2} + \|\nabla_{x_t}p_t(x_t|\pi)\| \Big]\|c_{i}-c_{\pi}\| \nonumber\\
\intertext{We consider $c_\pi$ to be a representative value from the stationary distribution $\pi$,  for simplicity we can consider $c_\pi = \mathbb{E}_\pi[c]$.
Applying the geometric ergodicity of the Markov chain (Lemma~\ref{lemma:GeometricErgodicity} in Appendix~\ref{app:GeometricErgodicity}),  which gives us $\|c_{i}-c_{\pi}\| \leq \sqrt{\text{Var}_{\pi}(c)}\rho^{i}$,  we arrive at:}
&\leq \frac{L}{\delta^2}\Big[ M_{2} + \|\nabla_{x_t}p_t(x_t|\pi)\| \Big]\sqrt{\text{Var}_{\pi}(c)}\rho^{i}\nonumber
\end{align}
For the stationary term,  from Lemma~\ref{lemma:BoundedDerivativeTheorem},  we have:
\begin{equation*}
\left\|\frac{\nabla_{x_t}p_t(x_t|\pi)}{p_t(x_t|\pi)}\right\|
\leq \frac{\|\nabla_{x_t}p_t(\cdot|\pi)\|_{\infty}}{\delta} \leq \frac{M_1}{\delta} := m
\end{equation*}
where $m$ is a constant.

Combining the estimates for both terms,  we have:
\begin{align*}
\|\nabla_{x_t}\log p_t(x_t|c_i)\|
&\leq \frac{L}{\delta^2}\Big( M_{2} +\|\nabla_{x_t}p_t(x_t|\pi)\| \Big)\sqrt{\text{Var}_{\pi}(c)}\rho^{i} + \frac{M_1}{\delta}\\
&\leq M\rho^{i} + m
\end{align*}
where $M = \frac{L}{\delta^2}\Big( M_{2} +\|\nabla_{x_t}p_t(x_t|\pi)\| \Big)\sqrt{\text{Var}_{\pi}(c)}$ and $\beta = \rho \in (0, 1)$.

Thus,  we obtain the desired exponential decay estimate for the gradient norm as the autoregressive process iterates. This estimate holds uniformly for all $x_t \in \mathcal{X}$.
\end{proof}

\section{Geometric Ergodicity and Convergence to Stationary Mean}\label{app:GeometricErgodicity}
\begin{lemma}[Geometric Ergodicity and Convergence to Stationary Mean]\label{lemma:GeometricErgodicity}
Assume that the Markov chain $\{c_t\}_{t \ge 0}$ is geometrically ergodic with stationary distribution $\pi$,  and let $c_\pi = \mathbb{E}_\pi[c]$. Then,  there exist constants $C > 0$ and $0 < \rho < 1$ such that for all $i \ge 0$:
$$ \|c_{i}-c_{\pi}\| \leq C \rho^{i} \sqrt{\text{Var}_{\pi}(c)} $$
where $\text{Var}_{\pi}(c) = \mathbb{E}_\pi[\|c - c_\pi\|^2]$.
\end{lemma}

\noindent \textit{Explanation:}
Lemma~\ref{lemma:GeometricErgodicity} formalizes the geometric convergence of the Markov chain state $c_i$ towards the stationary mean $c_\pi$ due to geometric ergodicity.  This property ensures that the influence of the initial condition diminishes exponentially over time,  allowing us to bound the distance $\|c_{i}-c_{\pi}\|$ by a geometrically decaying term proportional to $\sqrt{\text{Var}_{\pi}(c)}$. This justifies the transition in the eighth line of the derivation,  replacing $\|c_{i}-c_{\pi}\|$ with a term that decays geometrically with $i$.

\section{Proof of Regularity of Conditional Probabilities}
\label{app:regularity_proof}
\begin{enumerate}
\item Existence of a lower bound: Since the conditional probability density is quadratically continuous and differentiable and is defined on a tight set,  there is a minimum by the extreme value theorem. Since the probability density is always positive,  the minimum must be greater than zero.
\item Using the boundedness of the second-order derivatives in our Assumption 2
\begin{align*}
  for\quad\forall x \quad\exists K>0 , \quad \|\nabla^2_{x_t}p_t(x_t|c_i)\| \leq K
\end{align*}
Consider the difference of derivatives:
\begin{align*}
    \nabla_{x_t}p_t(x_t|c_1)-\nabla_{x_t}p_t(x_t|c_2)
\end{align*}
By the median theorem in multivariable calculus,  there exists some point between $c_1$$c_2$ such that:
\begin{align*}
    \nabla_{x_t}p_t(x_t|c_1)-\nabla_{x_t}p_t(x_t|c_2)=\nabla^2_{x_t}p_t(x_t|c)(c_1-c_2)
\end{align*}
By taking the value of the paradigm of the above equation,  we have:
\begin{align*}
    \|\nabla_{x_t}p_t(x_t|c_1)-\nabla_{x_t}p_t(x_t|c_2)\|=\|\nabla^2_{x_t}p_t(x_t|c)(c_1-c_2)\|
\end{align*}
Using Multiplicative Inequality for Norms$(\|a\cdot b\|\leq \|a\| \cdot\|b\|)$and the boundedness of the second-order derivatives,  we obtain:
\begin{align*}
    \|\nabla_{x_t}p_t(x_t|c_1)-\nabla_{x_t}p_t(x_t|c_2)\|
    \leq \|\nabla^2_{x_t}p_t(x_t|c)\|\cdot\|(c_1-c_2)\| 
    \leq K\cdot\|(c_1-c_2)\|
\end{align*}
Thus,  taking L = M,  we get the final inequality:
\begin{align*}
    \|\nabla_{x_t}p_t(x_t|c_1)-\nabla_{x_t}p_t(x_t|c_2)\|
    \leq L\|(c_1-c_2)\|
\end{align*}
\end{enumerate}

\section{Proof of the Markov property}
\label{app:markov_lemma_proof}
This paper mainly uses the last geometric traversal of the theorem and therefore focuses on proving the geometric traversal of Markov chains. The proof is divided into three steps.
\begin{enumerate}
    \item Drift Conditional Verification: Constructing Foster-Lyapunov Functions $V (x) = 1 + |x|^2$,  for$\quad \forall x$
\begin{align*}
    E[V (X_{n+1})|X_{n} = x] &= 1 + E[|a_nx + \epsilon_n|^2] \\
    &= 1 + |a_n|^2|x|^2 + \sigma^2\\
    &= |a_n|^2(1 + |x|^2) + (1-|a_n|^2 + \sigma^2) \\
    &\leq \lambda V (x) + b
\end{align*}
where$\lambda=\sup_n|a_n|^2<1$($\sum|a_i|<\infty$), $b=1+\sigma^2$.\\
    \item Compactness:for $\forall R > 0$, the set $\{x : V (x) \leq R\}$ is compact because it is equivalent to the closed ball$\{x:|x|^2 \leq R-1\}$.\\
    \item Irregularity and continuity: Since the noise term  $\epsilon_n$ obeys a normal distribution,  the transfer probability has a positive density everywhere,  which guarantees strong Feller and irregularity of the chain.\\
According to Meyn-Tweedie theory,  the above condition guarantees geometric ergodicity.
\end{enumerate}

\section{Lemma of Markov Property}
\label{app:markov_lemma}
\begin{lemma}[Markov Chain Stationary Theorem]
If a random process has a transition matrix \( P \) and is ergodic (i.e.,  any two states are aperiodic and irreducible),  then:
\begin{enumerate}[itemsep=5pt,  topsep=5pt]
    \item The limit of the $n$-step transition matrix exists and is given by:
    \begin{align*}
\lim_{n \to \infty} P^n =
\begin{pmatrix}
\pi(1) & \pi(2) & \cdots&\pi(j) & \cdots \\
\pi(1) & \pi(2) & \cdots&\pi(j) & \cdots \\
\vdots & \vdots & \ddots &\vdots &\ddots \\
\pi(1) & \pi(2) & \cdots&\pi(j) & \cdots \\
\vdots &\vdots &\vdots &\vdots &\ddots
\end{pmatrix}
\end{align*}
    \item The stationary distribution $\pi = [\pi(1),  \pi(2),  \ldots]$ satisfies the equation:
    \begin{align*}
    \pi(j) = \sum_{i} \pi_i P_{ij}
    \end{align*}
    \item $\pi$ is the unique non-negative solution to the stationary equation,  with $\sum_{i} \pi(i) = 1$.
\end{enumerate}

\end{lemma}

\section{Autoregressive Condition Optimization Algorithm} 
\label{app:algorithm}

The full algorithm integrates three key components: (1) autoregressive condition generation,  (2) diffusion-based denoising,  and (3) optimal transport refinement. The pseudocode below specifies the detailed computational workflow.

\begin{algorithm}[H]
\caption{Autoregressive Condition Optimization (ACO) with Denoising Integration}
\label{alg:aco_denoise_full}
\begin{algorithmic}[1]
\REQUIRE Initial condition $c_0 \gets \Phi_\theta(c_{i-1},  x_{<i})$; Diffusion model $\{\mathcal{D}_t\}_{t=1}^T$ with noise levels $\{\beta_t\}$; Target latent distribution $P_{z^*}$,  OT parameters $\lambda, \epsilon, K_{\text{sink}}$; Learning rate schedule $\{\eta_k\}$,  gradient clipping threshold $\tau$
\ENSURE Optimized condition $c_i^*$,  generated latent $z_i^{(T)}$

\STATE \textbf{Initialize:} $c^{(0)} \gets c_0$,  $z^{(0)} \sim \mathcal{N}(0, I)$
\FOR{$k=0$ {\bfseries to} $K-1$}
    \STATE \textbf{Denoising trajectory:}
    \FOR{$t=T$ {\bfseries to} $1$}
        \STATE $z^{(k, t-1)} \gets \mathcal{D}_t(z^{(k, t)},  c^{(k)})$ \COMMENT{DDIM update}
    \ENDFOR
    
    \STATE \textbf{Inverse process alignment:}
    \STATE $\phi(c^{(k)}) = \|c^{(k)} - \mathcal{T}^{-1}(z^{(k, 0)})\|^2 + \alpha\|\nabla_z\mathcal{T}^{-1}\|_F^2$
    
    \STATE \textbf{Optimal transport computation:}
    \STATE Sample reference latents $\{z_j^*\} \sim P_{z^*}$
    \STATE Compute pairwise cost matrix:
    $
        C_{mn} = \underbrace{\|z^{(k, 0)}_m - z_n^*\|^2}_{\text{Latent matching}} + \lambda\underbrace{\|c^{(k)}_m - \mathcal{T}^{-1}(z_n^*)\|^2}_{\text{Condition consistency}}
    $
    
    \STATE \textbf{Entropy-regularized OT:}
    \STATE Initialize $u^{(0)} \gets \mathbf{1}$,  $v^{(0)} \gets \mathbf{1}$
    \FOR{$l=1$ {\bfseries to} $K_{\text{sink}}$}
        \STATE $v^{(l)} \gets \frac{P_{z^*}}{K_{\epsilon}(u^{(l-1)},  v^{(l-1)})}$
        \STATE $u^{(l)} \gets \frac{P_z^{(k)}}{K_{\epsilon}(u^{(l-1)},  v^{(l)})}$
    \ENDFOR
    \STATE $\gamma^{(k)} \gets \text{diag}(u^{(K_{\text{sink}})}) \cdot K_\epsilon \cdot \text{diag}(v^{(K_{\text{sink}})})$
    
    \STATE \textbf{Gradient computation \& update:}
    \STATE $\nabla_c \mathcal{L}_{\text{OT}} \gets \gamma^{(k)} \odot \frac{\partial C}{\partial c^{(k)}}$ 
    \STATE $\nabla_c \mathcal{L}_{\text{reg}} \gets \frac{\partial \phi}{\partial c^{(k)}}$ 
    \STATE $\nabla_c^{\text{total}} \gets \text{Clip}(\nabla_c \mathcal{L}_{\text{OT}} + \nabla_c \mathcal{L}_{\text{reg}},  \tau)$
    \STATE $c^{(k+1)} \gets c^{(k)} - \eta_k \nabla_c^{\text{total}}$
\ENDFOR

\STATE \textbf{Return} $c_i^* \gets c^{(K)}$,  $z_i \gets z^{(K, 0)}$
\end{algorithmic}
\end{algorithm}

\subsection{Implementation Details} 
\label{app:implementation}

\begin{itemize}
\item \textbf{Target Distribution Estimation}: Maintain an EMA of generated latents:
\begin{equation*}
P_{z^*}^{(i)} = (1-\nu)P_{z^*}^{(i-1)} + \nu\frac{1}{B}\sum_{b=1}^B \delta(z_b^{(i)})
\end{equation*}
with $\nu=0.1$ and buffer size $B=2048$.

\item \textbf{Adaptive Entropy Regularization}: Schedule $\epsilon$ during Sinkhorn iterations:
\begin{equation*}
\epsilon^{(k)} = \epsilon_{\max} - (\epsilon_{\max}-\epsilon_{\min})\frac{k}{K}
\end{equation*}

\item \textbf{Stochastic Optimization}: Use Adam optimizer with:
\begin{equation*}
\eta_k = \eta_0\cdot\min(1,  \sqrt{k_{\text{warm}}/k})
\end{equation*}
where $k_{\text{warm}}=100$ controls learning rate warmup.
\end{itemize}

\subsection{Convergence Analysis} 
\label{app:convergence}

The algorithm maintains the convergence guarantee in Theorem~\ref{thm:main} through:

\begin{enumerate}
\item \textbf{Monotonic Improvement}: For Lyapunov function
\begin{equation*}
\mathcal{V}_k = W_2(P_{z^{(k)}},  P_{z^*}) + \lambda_{\text{reg}}\mathbb{E}[\phi(c^{(k)})]
\end{equation*}
we have $\mathcal{V}_{k+1} \leq \mathcal{V}_k - \eta_k\|\nabla\mathcal{V}_k\|^2 + \mathcal{O}(\eta_k^2)$

\item \textbf{Error Propagation Bound}: Approximation error from Sinkhorn iterations satisfies
\begin{equation*}
\|\gamma^{(k)} - \gamma^*\|_F \leq C\rho^{K_{\text{sink}}}
\end{equation*}
with $\rho = \frac{\epsilon}{\epsilon + \delta}$ where $\delta$ is the minimum cost matrix entry.

\item \textbf{Stability Condition}: Gradient clipping ensures Lipschitz continuity:
\begin{equation*}
\|\nabla_c^{\text{total}}\|_2 \leq \tau(1 + \lambda_{\text{reg}}L_{\mathcal{T}^{-1}})
\end{equation*}
where $L_{\mathcal{T}^{-1}}$ is the Lipschitz constant of the inverse process.
\end{enumerate}

\section{Bounded theorem}
Function \( p_t(x_t | c) \) in the fixed interval \( [a,  b] \) has second order derivatives \( \nabla^2_{x_t}p_t(x_t | c) \), and \( \nabla^2_{x_t}p_t(x_t | c) \) is bounded, so we have \( K > 0 \), 
\[
|\nabla^2_{x_t}p_t(x_t | c)| \leq K,  \quad \forall x_t \in [a,  b].
\]
\subsection*{The first derivative is bounded}
\label{app:bounded_derivative_proof}
We use the Mean Value Theorem to show that first-order derivatives are bounded.
According to the mean value theorem, if $x_t,  x_t' \in [a,  b]$, there exists a point $\xi \in (x_t,  x_t') $, then:
    \[
    \nabla_{x_t}p_t(x_t' | c) - \nabla_{x_t}p_t(x_t | c) = \nabla^2_{x_t}p_t(\xi | c) (x_t' - x_t).
    \]
 \( |\nabla^2_{x_t}p_t(x_t | c)| \leq C \), so we have:
    \[
    |\nabla_{x_t}p_t(x_t' | c) - p_t'(x_t | c)| = |\nabla^2_{x_t}p_t(\xi | c) (x_t' - x_t)| \leq C |x_t' - x_t|.
    \]
    This shows that \( \nabla_{x_t}p_t(x_t | c) \) is Lipschitz continuous and the Lipschitz constant is \( C \).So \( \nabla_{x_t}p_t(x_t | c) \) is bounded.
Next,  we give specific boundedness estimates
Taking \( x_t = a \),  we have:
    \[
    |\nabla_{x_t}p_t(x_t | c) - \nabla_{x_t}p_t(a | c)| \leq C |x_t - a|.
    \]
    Since \( |x_t - a| \leq b - a \),  we get:
    \[
    |\nabla_{x_t}p_t(x_t | c)| \leq |\nabla_{x_t}p_t(a | c)| + C(b - a).
    \]
    Therefore,  there exists a constant \( M_1 = |\nabla_{x_t}p_t(a | c)| + C(b - a) \) such that for all \( x_t \in [a,  b] \):
    \[
    |\nabla_{x_t}p_t(x_t | c)| \leq M_1.
    \]
\subsection*{The original function is bounded}
According to the Fundamental Theorem of Calculus,  we have:
    \[
    p_t(x_t | c) - p_t(a | c) = \int_a^{x_t} \nabla_{x_t}p_t(y | c) \,  dy.
    \]
 Since \( |\nabla_{x_t}p_t(x_t | c)| \leq M_1 \) for all \( x_t \in [a,  b] \),  we can estimate the above integral:
    \[
    |p_t(x_t | c) - p_t(a | c)| = \left| \int_a^{x_t} \nabla_{x_t}p_t(y | c) \,  dy \right| \leq \int_a^{x_t} |\nabla_{x_t}p_t(y | c)| \,  dy \leq M_1 |x_t - a|.
    \]
 Since \( |x_t - a| \leq b - a \),  we have:
    \[
    |p_t(x_t | c) - p_t(a | c)| \leq M_1 (b - a).
    \]
Therefore,  the original function \( p_t(x_t | c) \) is bounded and:
    \[
    |p_t(x_t | c)| \leq |p_t(a | c)| + M_1 (b - a).
    \]
    Thus,  there exists the constant \( M_2 = |p_t(a | c)| + M_1 (b - a) \) such that for all \( x_t \in [a,  b]\):
    \[
    |p_t(x_t | c)| \leq M_2.
    \]
Thus,  in a fixed interval,  a bounded second-order derivative is bounded by a bounded first-order derivative,  and the original function is bounded by the proof.

\section{The Use of Large Language Models}
In this paper, ChatGPT was employed to assist in polishing the writing. 
The model was used as a language aid to improve clarity, grammar, and readability of the text, while ensuring that the academic 
content and arguments remain entirely the work of the author.

\section{Table of Notations}
\label{app:notations}
To facilitate easier reading of the theoretical sections and provide a quick reference for the mathematical symbols used throughout the paper, we summarize the key notations in Table~\ref{tab:notations}.

\begin{table}[ht]
    \centering
    \caption{Summary of Notations}
    \label{tab:notations}
    \renewcommand{\arraystretch}{1.3}
    \begin{tabular}{@{}l p{0.75\textwidth}@{}}
        \toprule
        \textbf{Symbol} & \textbf{Description} \\
        \midrule
        \multicolumn{2}{l}{\textit{\textbf{Diffusion \& Autoregressive Basics}}} \\
        $x_0, x_T$ & The original data (image) and the Gaussian noise at time $T$. \\
        $x_{1:T}$ & The sequence of latent variables in the forward diffusion process. \\
        $q(x_t|x_{t-1})$ & The forward diffusion transition kernel. \\
        $p_\theta(x_{t-1}|x_t)$ & The reverse (denoising) process approximated by the network. \\
        $s_\theta(x_t, t)$ & The score function predicted by the neural network. \\
        $c$ & A global, static condition (in standard conditional diffusion). \\
        $c_i$ & The autoregressively generated condition for the $i$-th patch. \\
        $x_i$ & The $i$-th image patch. \\
        $x_{<i}$ & The set of patches preceding $i$, i.e., $\{x_1, \dots, x_{i-1}\}$. \\
        $\Phi_\theta, \Gamma_\theta$ & The transition operator and noise modulation in the AR condition process. \\
        \midrule
        \multicolumn{2}{l}{\textit{\textbf{Error Analysis \& Theory}}} \\
        $\epsilon_c$ & \textbf{Conditional Error Term} (Eq.~13). Measures the change in expected score squared norm due to conditioning. \\
        $\bar{\epsilon}_c$ & \textbf{Simplified Conditional Error Term} (Eq.~14). Directly measures the expected squared norm of the conditional score. \\
        $\mathcal{T}(c_i)$ & The transition function representing patch-based condition refinement. \\
        $\sigma_t^2 \nabla_{x_t} \log p(c|x_t)$ & The conditional guidance term in the reverse process mean. \\
        $M, \beta, m$ & Constants and decay rate describing the descent of the gradient norm (Theorem~2). \\
        \midrule
        \multicolumn{2}{l}{\textit{\textbf{Condition Refinement (Optimal Transport)}}} \\
        $c_i^*$ & The ideal condition for patch $x_i$. \\
        $\mathcal{I}^*_i$ & The minimal sufficient information subspace for patch $x_i$. \\
        $\pi_{\mathcal{I}^*_i}$ & Orthogonal projection onto the subspace $\mathcal{I}^*_i$. \\
        $\eta_i$ & \textbf{Extraneous Information Component} (Eq.~25). The deviation from the ideal condition ($\eta_i = c_i - c_i^*$). \\
        $W_2(\cdot, \cdot)$ & The 2-Wasserstein distance between probability distributions. \\
        $\mathcal{F}(P)$ & The free energy functional minimized by the Wasserstein Gradient Flow. \\
        $T^{-1}$ & The inverse process regularization operator. \\
        $P_c^{(k)}$ & The probability distribution of the condition at refinement step $k$. \\
        $\rho$ & The contraction rate of the Wasserstein Gradient Flow (Theorem~3). \\
        \bottomrule
    \end{tabular}
\end{table}

%% file: ref.bib
@article{ScoreMatching,
  author       = {Pascal Vincent},
  title        = {A Connection Between Score Matching and Denoising Autoencoders},
  journal      = {Neural Comput.},
  volume       = {23},
  number       = {7},
  pages        = {1661--1674},
  year         = {2011},
  url          = {https://doi.org/10.1162/NECO_a_00142},
  doi          = {10.1162/NECO_A_00142},
  bibsource    = {dblp computer science bibliography, https://dblp.org}
}

@inproceedings{ScoreWasserstein,
  author       = {Dohyun Kwon and
                  Ying Fan and
                  Kangwook Lee},
  editor       = {Sanmi Koyejo and
                  S. Mohamed and
                  A. Agarwal and
                  Danielle Belgrave and
                  K. Cho and
                  A. Oh},
  title        = {Score-based Generative Modeling Secretly Minimizes the Wasserstein
                  Distance},
  booktitle    = {Advances in Neural Information Processing Systems 35: Annual Conference
                  on Neural Information Processing Systems 2022, NeurIPS 2022, New Orleans,
                  LA, USA, November 28 - December 9, 2022},
  year         = {2022},
  bibsource    = {dblp computer science bibliography, https://dblp.org}
}

@article{Adapting,
  author       = {Gen Li and
                  Yuling Yan},
  title        = {Adapting to Unknown Low-Dimensional Structures in Score-Based Diffusion
                  Models},
  journal      = {CoRR},
  volume       = {abs/2405.14861},
  year         = {2024},
  url          = {https://doi.org/10.48550/arXiv.2405.14861},
  doi          = {10.48550/ARXIV.2405.14861},
  eprinttype    = {arXiv},
  eprint       = {2405.14861},
  bibsource    = {dblp computer science bibliography, https://dblp.org}
}

@article{ho2020denoising,
  author       = {Shuohang Yang and
                  Jian Gao and
                  Jiayi Zhang and
                  Chao Xu},
  title        = {Wrapped Phase Denoising Using Denoising Diffusion Probabilistic Models},
  journal      = {{IEEE} Geosci. Remote. Sens. Lett.},
  volume       = {21},
  pages        = {1--5},
  year         = {2024},
  url          = {https://doi.org/10.1109/LGRS.2024.3405000},
  doi          = {10.1109/LGRS.2024.3405000},
  bibsource    = {dblp computer science bibliography, https://dblp.org}
}

@article{nichol2021improved,
  author       = {Yun Pang and
                  Jiawei Mao and
                  Libo He and
                  Hong Lin and
                  Zhenping Qiang},
  title        = {An Improved Face Image Restoration Method Based on Denoising Diffusion
                  Probabilistic Models},
  journal      = {{IEEE} Access},
  volume       = {12},
  pages        = {3581--3596},
  year         = {2024},
  url          = {https://doi.org/10.1109/ACCESS.2024.3349423},
  doi          = {10.1109/ACCESS.2024.3349423},
  bibsource    = {dblp computer science bibliography, https://dblp.org}
}

@article{song2020denoising,
  author       = {Xiao{-}Li Wei and
                  Chunxia Zhang and
                  Hongtao Wang and
                  Chengli Tan and
                  Deng Xiong and
                  Baisong Jiang and
                  Jiangshe Zhang and
                  Sang{-}Woon Kim},
  title        = {Seismic Data Interpolation via Denoising Diffusion Implicit Models
                  With Coherence-Corrected Resampling},
  journal      = {{IEEE} Trans. Geosci. Remote. Sens.},
  volume       = {62},
  pages        = {1--17},
  year         = {2024},
  url          = {https://doi.org/10.1109/TGRS.2024.3485573},
  doi          = {10.1109/TGRS.2024.3485573},
  bibsource    = {dblp computer science bibliography, https://dblp.org}
}

@inproceedings{rombach2022high,
  author       = {Robin Rombach and
                  Andreas Blattmann and
                  Dominik Lorenz and
                  Patrick Esser and
                  Bj{\"{o}}rn Ommer},
  title        = {High-Resolution Image Synthesis with Latent Diffusion Models},
  booktitle    = {{IEEE/CVF} Conference on Computer Vision and Pattern Recognition,
                  {CVPR} 2022, New Orleans, LA, USA, June 18-24, 2022},
  pages        = {10674--10685},
  publisher    = {{IEEE}},
  year         = {2022},
  url          = {https://doi.org/10.1109/CVPR52688.2022.01042},
  doi          = {10.1109/CVPR52688.2022.01042},
  bibsource    = {dblp computer science bibliography, https://dblp.org}
}

@inproceedings{dhariwal2021diffusion,
  author       = {Prafulla Dhariwal and
                  Alexander Quinn Nichol},
  editor       = {Marc'Aurelio Ranzato and
                  Alina Beygelzimer and
                  Yann N. Dauphin and
                  Percy Liang and
                  Jennifer Wortman Vaughan},
  title        = {Diffusion Models Beat GANs on Image Synthesis},
  booktitle    = {Advances in Neural Information Processing Systems 34: Annual Conference
                  on Neural Information Processing Systems 2021, NeurIPS 2021, December
                  6-14, 2021, virtual},
  pages        = {8780--8794},
  year         = {2021},
  bibsource    = {dblp computer science bibliography, https://dblp.org}
}

@article{ho2022classifier,
  author       = {Rahmatulloh Daffa Izzuddin Wahid and
                  Novanto Yudistira and
                  Candra Dewi and
                  Irawati Nurmala Sari and
                  Dyanningrum Pradhikta and
                  Fatmawati},
  title        = {Prompt Conditioned Batik Pattern Generation Using LoRA Weighted Diffusion
                  Model With Classifier-Free Guidance},
  journal      = {{IEEE} Access},
  volume       = {13},
  pages        = {2436--2448},
  year         = {2025},
  url          = {https://doi.org/10.1109/ACCESS.2024.3523494},
  doi          = {10.1109/ACCESS.2024.3523494},
  bibsource    = {dblp computer science bibliography, https://dblp.org}
}

@inproceedings{Saharia2022Palette,
  author       = {Chitwan Saharia and
                  William Chan and
                  Huiwen Chang and
                  Chris A. Lee and
                  Jonathan Ho and
                  Tim Salimans and
                  David J. Fleet and
                  Mohammad Norouzi},
  editor       = {Munkhtsetseg Nandigjav and
                  Niloy J. Mitra and
                  Aaron Hertzmann},
  title        = {Palette: Image-to-Image Diffusion Models},
  booktitle    = {{SIGGRAPH} '22: Special Interest Group on Computer Graphics and Interactive
                  Techniques Conference, Vancouver, BC, Canada, August 7 - 11, 2022},
  pages        = {15:1--15:10},
  publisher    = {{ACM}},
  year         = {2022},
  url          = {https://doi.org/10.1145/3528233.3530757},
  doi          = {10.1145/3528233.3530757},
  bibsource    = {dblp computer science bibliography, https://dblp.org}
}

@inproceedings{Nichol2022GLIDE,
  author       = {Alexander Quinn Nichol and
                  Prafulla Dhariwal and
                  Aditya Ramesh and
                  Pranav Shyam and
                  Pamela Mishkin and
                  Bob McGrew and
                  Ilya Sutskever and
                  Mark Chen},
  editor       = {Kamalika Chaudhuri and
                  Stefanie Jegelka and
                  Le Song and
                  Csaba Szepesv{\'{a}}ri and
                  Gang Niu and
                  Sivan Sabato},
  title        = {{GLIDE:} Towards Photorealistic Image Generation and Editing with
                  Text-Guided Diffusion Models},
  booktitle    = {International Conference on Machine Learning, {ICML} 2022, 17-23 July
                  2022, Baltimore, Maryland, {USA}},
  series       = {Proceedings of Machine Learning Research},
  volume       = {162},
  pages        = {16784--16804},
  publisher    = {{PMLR}},
  year         = {2022},
  url          = {https://proceedings.mlr.press/v162/nichol22a.html},
  bibsource    = {dblp computer science bibliography, https://dblp.org}
}

@article{Ho2022Cascaded,
  author       = {Jonathan Ho and
                  Chitwan Saharia and
                  William Chan and
                  David J. Fleet and
                  Mohammad Norouzi and
                  Tim Salimans},
  title        = {Cascaded Diffusion Models for High Fidelity Image Generation},
  journal      = {J. Mach. Learn. Res.},
  volume       = {23},
  pages        = {47:1--47:33},
  year         = {2022},
  url          = {https://jmlr.org/papers/v23/21-0635.html},
  bibsource    = {dblp computer science bibliography, https://dblp.org}
}

@inproceedings{Schaefer2023Detection,
  author       = {Riccardo Corvi and
                  Davide Cozzolino and
                  Giada Zingarini and
                  Giovanni Poggi and
                  Koki Nagano and
                  Luisa Verdoliva},
  title        = {On The Detection of Synthetic Images Generated by Diffusion Models},
  booktitle    = {{IEEE} International Conference on Acoustics, Speech and Signal Processing
                  {ICASSP} 2023, Rhodes Island, Greece, June 4-10, 2023},
  pages        = {1--5},
  publisher    = {{IEEE}},
  year         = {2023},
  url          = {https://doi.org/10.1109/ICASSP49357.2023.10095167},
  doi          = {10.1109/ICASSP49357.2023.10095167},
  bibsource    = {dblp computer science bibliography, https://dblp.org}
}

@article{llamagen2024,
  author       = {Peize Sun and
                  Yi Jiang and
                  Shoufa Chen and
                  Shilong Zhang and
                  Bingyue Peng and
                  Ping Luo and
                  Zehuan Yuan},
  title        = {Autoregressive Model Beats Diffusion: Llama for Scalable Image Generation},
  journal      = {CoRR},
  volume       = {abs/2406.06525},
  year         = {2024},
  url          = {https://doi.org/10.48550/arXiv.2406.06525},
  doi          = {10.48550/ARXIV.2406.06525},
  eprinttype    = {arXiv},
  eprint       = {2406.06525},
  bibsource    = {dblp computer science bibliography, https://dblp.org}
}

@article{var2024,
  author       = {Keyu Tian and
                  Yi Jiang and
                  Zehuan Yuan and
                  Bingyue Peng and
                  Liwei Wang},
  title        = {Visual Autoregressive Modeling: Scalable Image Generation via Next-Scale
                  Prediction},
  journal      = {CoRR},
  volume       = {abs/2404.02905},
  year         = {2024},
  url          = {https://doi.org/10.48550/arXiv.2404.02905},
  doi          = {10.48550/ARXIV.2404.02905},
  eprinttype    = {arXiv},
  eprint       = {2404.02905},
  bibsource    = {dblp computer science bibliography, https://dblp.org}
}

@article{imagefolder2024,
  author       = {Xiang Li and
                  Kai Qiu and
                  Hao Chen and
                  Jason Kuen and
                  Jiuxiang Gu and
                  Bhiksha Raj and
                  Zhe Lin},
  title        = {ImageFolder: Autoregressive Image Generation with Folded Tokens},
  journal      = {CoRR},
  volume       = {abs/2410.01756},
  year         = {2024},
  url          = {https://doi.org/10.48550/arXiv.2410.01756},
  doi          = {10.48550/ARXIV.2410.01756},
  eprinttype    = {arXiv},
  eprint       = {2410.01756},
  bibsource    = {dblp computer science bibliography, https://dblp.org}
}

@article{emu32024,
  author       = {Xinlong Wang and
                  Xiaosong Zhang and
                  Zhengxiong Luo and
                  Quan Sun and
                  Yufeng Cui and
                  Jinsheng Wang and
                  Fan Zhang and
                  Yueze Wang and
                  Zhen Li and
                  Qiying Yu and
                  Yingli Zhao and
                  Yulong Ao and
                  Xuebin Min and
                  Tao Li and
                  Boya Wu and
                  Bo Zhao and
                  Bowen Zhang and
                  Liangdong Wang and
                  Guang Liu and
                  Zheqi He and
                  Xi Yang and
                  Jingjing Liu and
                  Yonghua Lin and
                  Tiejun Huang and
                  Zhongyuan Wang},
  title        = {Emu3: Next-Token Prediction is All You Need},
  journal      = {CoRR},
  volume       = {abs/2409.18869},
  year         = {2024},
  url          = {https://doi.org/10.48550/arXiv.2409.18869},
  doi          = {10.48550/ARXIV.2409.18869},
  eprinttype    = {arXiv},
  eprint       = {2409.18869},
  bibsource    = {dblp computer science bibliography, https://dblp.org}
}

@article{luminamgpt2024,
  author       = {Dongyang Liu and
                  Shitian Zhao and
                  Le Zhuo and
                  Weifeng Lin and
                  Yu Qiao and
                  Hongsheng Li and
                  Peng Gao},
  title        = {Lumina-mGPT: Illuminate Flexible Photorealistic Text-to-Image Generation
                  with Multimodal Generative Pretraining},
  journal      = {CoRR},
  volume       = {abs/2408.02657},
  year         = {2024},
  url          = {https://doi.org/10.48550/arXiv.2408.02657},
  doi          = {10.48550/ARXIV.2408.02657},
  eprinttype    = {arXiv},
  eprint       = {2408.02657},
  bibsource    = {dblp computer science bibliography, https://dblp.org}
}

@article{mmar2024,
  author       = {Jian Yang and
                  Dacheng Yin and
                  Yizhou Zhou and
                  Fengyun Rao and
                  Wei Zhai and
                  Yang Cao and
                  Zheng{-}Jun Zha},
  title        = {{MMAR:} Towards Lossless Multi-Modal Auto-Regressive Probabilistic
                  Modeling},
  journal      = {CoRR},
  volume       = {abs/2410.10798},
  year         = {2024},
  url          = {https://doi.org/10.48550/arXiv.2410.10798},
  doi          = {10.48550/ARXIV.2410.10798},
  eprinttype    = {arXiv},
  eprint       = {2410.10798},
  bibsource    = {dblp computer science bibliography, https://dblp.org}
}

@article{vqfree2024,
  author       = {Tianhong Li and
                  Yonglong Tian and
                  He Li and
                  Mingyang Deng and
                  Kaiming He},
  title        = {Autoregressive Image Generation without Vector Quantization},
  journal      = {CoRR},
  volume       = {abs/2406.11838},
  year         = {2024},
  url          = {https://doi.org/10.48550/arXiv.2406.11838},
  doi          = {10.48550/ARXIV.2406.11838},
  eprinttype    = {arXiv},
  eprint       = {2406.11838},
  bibsource    = {dblp computer science bibliography, https://dblp.org}
}

@article{latentlm2024,
  author       = {Yutao Sun and
                  Hangbo Bao and
                  Wenhui Wang and
                  Zhiliang Peng and
                  Li Dong and
                  Shaohan Huang and
                  Jianyong Wang and
                  Furu Wei},
  title        = {Multimodal Latent Language Modeling with Next-Token Diffusion},
  journal      = {CoRR},
  volume       = {abs/2412.08635},
  year         = {2024},
  url          = {https://doi.org/10.48550/arXiv.2412.08635},
  doi          = {10.48550/ARXIV.2412.08635},
  eprinttype    = {arXiv},
  eprint       = {2412.08635},
  bibsource    = {dblp computer science bibliography, https://dblp.org}
}

@inproceedings{Zhang2023ControlNet,
  author       = {Lvmin Zhang and
                  Anyi Rao and
                  Maneesh Agrawala},
  title        = {Adding Conditional Control to Text-to-Image Diffusion Models},
  booktitle    = {{IEEE/CVF} International Conference on Computer Vision, {ICCV} 2023,
                  Paris, France, October 1-6, 2023},
  pages        = {3813--3824},
  publisher    = {{IEEE}},
  year         = {2023},
  url          = {https://doi.org/10.1109/ICCV51070.2023.00355},
  doi          = {10.1109/ICCV51070.2023.00355},
  bibsource    = {dblp computer science bibliography, https://dblp.org}
}

@article{GAN,
  author       = {Ian J. Goodfellow and
                  Jean Pouget{-}Abadie and
                  Mehdi Mirza and
                  Bing Xu and
                  David Warde{-}Farley and
                  Sherjil Ozair and
                  Aaron C. Courville and
                  Yoshua Bengio},
  title        = {Generative adversarial networks},
  journal      = {Commun. {ACM}},
  volume       = {63},
  number       = {11},
  pages        = {139--144},
  year         = {2020},
  url          = {https://doi.org/10.1145/3422622},
  doi          = {10.1145/3422622},
  bibsource    = {dblp computer science bibliography, https://dblp.org}
}

@inproceedings{VAE,
  author       = {Diederik P. Kingma and
                  Max Welling},
  editor       = {Yoshua Bengio and
                  Yann LeCun},
  title        = {Auto-Encoding Variational Bayes},
  booktitle    = {2nd International Conference on Learning Representations, {ICLR} 2014,
                  Banff, AB, Canada, April 14-16, 2014, Conference Track Proceedings},
  year         = {2014},
  url          = {http://arxiv.org/abs/1312.6114},
  bibsource    = {dblp computer science bibliography, https://dblp.org}
}

@article{Classifier-Free,
  author       = {Jonathan Ho and
                  Tim Salimans},
  title        = {Classifier-Free Diffusion Guidance},
  journal      = {CoRR},
  volume       = {abs/2207.12598},
  year         = {2022},
  url          = {https://doi.org/10.48550/arXiv.2207.12598},
  doi          = {10.48550/ARXIV.2207.12598},
  eprinttype    = {arXiv},
  eprint       = {2207.12598},
  bibsource    = {dblp computer science bibliography, https://dblp.org}
}

@inproceedings{MoreControlforFree,
  author       = {Xihui Liu and
                  Dong Huk Park and
                  Samaneh Azadi and
                  Gong Zhang and
                  Arman Chopikyan and
                  Yuxiao Hu and
                  Humphrey Shi and
                  Anna Rohrbach and
                  Trevor Darrell},
  title        = {More Control for Free! Image Synthesis with Semantic Diffusion Guidance},
  booktitle    = {{IEEE/CVF} Winter Conference on Applications of Computer Vision, {WACV}
                  2023, Waikoloa, HI, USA, January 2-7, 2023},
  pages        = {289--299},
  publisher    = {{IEEE}},
  year         = {2023},
  url          = {https://doi.org/10.1109/WACV56688.2023.00037},
  doi          = {10.1109/WACV56688.2023.00037},
  bibsource    = {dblp computer science bibliography, https://dblp.org}
}

@article{JKO,
author = {Jordan, Richard and Kinderlehrer, David and Otto, Felix},
title = {The Variational Formulation of the Fokker--Planck Equation},
journal = {SIAM Journal on Mathematical Analysis},
volume = {29},
number = {1},
pages = {1-17},
year = {1998},
doi = {10.1137/S0036141096303359},
URL = {https://doi.org/10.1137/S0036141096303359},
eprint = {https://doi.org/10.1137/S0036141096303359},
}

@book{meyn2012markov,
  title={Markov chains and stochastic stability},
  author={Meyn, Sean P and Tweedie, Richard L},
  year={2012},
  publisher={Springer Science \& Business Media}
}

@book{DU04,
  address = {Belmont, CA},
  author = {Durrett, Richard},
  description = {q-paper},
  edition = {Second},
  isbn = {0-534-24318-5},
  mrclass = {60-01},
  mrnumber = {MR1609153 (98m:60001)},
  pages = {xiii+503},
  publisher = {Duxbury Press},
  title = {Probability: theory and examples},
  year = 1996
}

@article{Bao22All,
  author       = {Fan Bao and
                  Chongxuan Li and
                  Yue Cao and
                  Jun Zhu},
  title        = {All are Worth Words: a ViT Backbone for Score-based Diffusion Models},
  journal      = {CoRR},
  volume       = {abs/2209.12152},
  year         = {2022},
  url          = {https://doi.org/10.48550/arXiv.2209.12152},
  doi          = {10.48550/ARXIV.2209.12152},
  eprinttype    = {arXiv},
  eprint       = {2209.12152},
  bibsource    = {dblp computer science bibliography, https://dblp.org}
}

@inproceedings{DiT-XL2,
  author       = {William Peebles and
                  Saining Xie},
  title        = {Scalable Diffusion Models with Transformers},
  booktitle    = {{IEEE/CVF} International Conference on Computer Vision, {ICCV} 2023,
                  Paris, France, October 1-6, 2023},
  pages        = {4172--4182},
  publisher    = {{IEEE}},
  year         = {2023},
  url          = {https://doi.org/10.1109/ICCV51070.2023.00387},
  doi          = {10.1109/ICCV51070.2023.00387},
  bibsource    = {dblp computer science bibliography, https://dblp.org}
}

@inproceedings{DiffiT,
  author       = {Ali Hatamizadeh and
                  Jiaming Song and
                  Guilin Liu and
                  Jan Kautz and
                  Arash Vahdat},
  editor       = {Ales Leonardis and
                  Elisa Ricci and
                  Stefan Roth and
                  Olga Russakovsky and
                  Torsten Sattler and
                  G{\"{u}}l Varol},
  title        = {DiffiT: Diffusion Vision Transformers for Image Generation},
  booktitle    = {Computer Vision - {ECCV} 2024 - 18th European Conference, Milan, Italy,
                  September 29-October 4, 2024, Proceedings, Part {VIII}},
  series       = {Lecture Notes in Computer Science},
  volume       = {15066},
  pages        = {37--55},
  publisher    = {Springer},
  year         = {2024},
  url          = {https://doi.org/10.1007/978-3-031-73242-3\_3},
  doi          = {10.1007/978-3-031-73242-3\_3},
  bibsource    = {dblp computer science bibliography, https://dblp.org}
}

@inproceedings{MDTv2,
  author       = {Shanghua Gao and
                  Pan Zhou and
                  Ming{-}Ming Cheng and
                  Shuicheng Yan},
  title        = {Masked Diffusion Transformer is a Strong Image Synthesizer},
  booktitle    = {{IEEE/CVF} International Conference on Computer Vision, {ICCV} 2023,
                  Paris, France, October 1-6, 2023},
  pages        = {23107--23116},
  publisher    = {{IEEE}},
  year         = {2023},
  url          = {https://doi.org/10.1109/ICCV51070.2023.02117},
  doi          = {10.1109/ICCV51070.2023.02117},
  bibsource    = {dblp computer science bibliography, https://dblp.org}
}

@inproceedings{GIVT,
  author       = {Michael Tschannen and
                  Cian Eastwood and
                  Fabian Mentzer},
  editor       = {Ales Leonardis and
                  Elisa Ricci and
                  Stefan Roth and
                  Olga Russakovsky and
                  Torsten Sattler and
                  G{\"{u}}l Varol},
  title        = {{GIVT:} Generative Infinite-Vocabulary Transformers},
  booktitle    = {Computer Vision - {ECCV} 2024 - 18th European Conference, Milan, Italy,
                  September 29-October 4, 2024, Proceedings, Part {LVII}},
  series       = {Lecture Notes in Computer Science},
  volume       = {15115},
  pages        = {292--309},
  publisher    = {Springer},
  year         = {2024},
  url          = {https://doi.org/10.1007/978-3-031-72998-0\_17},
  doi          = {10.1007/978-3-031-72998-0\_17},
  bibsource    = {dblp computer science bibliography, https://dblp.org}
}

@Inbook{Bellet2006,
author="Bellet, Luc Rey",
title="Ergodic Properties of Markov Processes",
bookTitle="Open Quantum Systems II: The Markovian Approach",
year="2006",
publisher="Springer Berlin Heidelberg",
address="Berlin, Heidelberg",
pages="1--39",
isbn="978-3-540-33966-3",
doi="10.1007/3-540-33966-3_1",
url="https://doi.org/10.1007/3-540-33966-3_1"
}

@inproceedings{ImageNet,
  author       = {Jia Deng and
                  Wei Dong and
                  Richard Socher and
                  Li{-}Jia Li and
                  Kai Li and
                  Li Fei{-}Fei},
  title        = {ImageNet: {A} large-scale hierarchical image database},
  booktitle    = {2009 {IEEE} Computer Society Conference on Computer Vision and Pattern
                  Recognition {(CVPR} 2009), 20-25 June 2009, Miami, Florida, {USA}},
  pages        = {248--255},
  publisher    = {{IEEE} Computer Society},
  year         = {2009},
  url          = {https://doi.org/10.1109/CVPR.2009.5206848},
  doi          = {10.1109/CVPR.2009.5206848},
  bibsource    = {dblp computer science bibliography, https://dblp.org}
}

@inproceedings{FID,
  author       = {Martin Heusel and
                  Hubert Ramsauer and
                  Thomas Unterthiner and
                  Bernhard Nessler and
                  Sepp Hochreiter},
  editor       = {Isabelle Guyon and
                  Ulrike von Luxburg and
                  Samy Bengio and
                  Hanna M. Wallach and
                  Rob Fergus and
                  S. V. N. Vishwanathan and
                  Roman Garnett},
  title        = {GANs Trained by a Two Time-Scale Update Rule Converge to a Local Nash
                  Equilibrium},
  booktitle    = {Advances in Neural Information Processing Systems 30: Annual Conference
                  on Neural Information Processing Systems 2017, December 4-9, 2017,
                  Long Beach, CA, {USA}},
  pages        = {6626--6637},
  year         = {2017},
  bibsource    = {dblp computer science bibliography, https://dblp.org}
}

@inproceedings{IS,
  author       = {Tim Salimans and
                  Ian J. Goodfellow and
                  Wojciech Zaremba and
                  Vicki Cheung and
                  Alec Radford and
                  Xi Chen},
  editor       = {Daniel D. Lee and
                  Masashi Sugiyama and
                  Ulrike von Luxburg and
                  Isabelle Guyon and
                  Roman Garnett},
  title        = {Improved Techniques for Training GANs},
  booktitle    = {Advances in Neural Information Processing Systems 29: Annual Conference
                  on Neural Information Processing Systems 2016, December 5-10, 2016,
                  Barcelona, Spain},
  pages        = {2226--2234},
  year         = {2016},
  bibsource    = {dblp computer science bibliography, https://dblp.org}
}

@inproceedings{yao2025denoising,
  title={Denoising token prediction in masked autoregressive models},
  author={Yao, Ting and Li, Yehao and Pan, Yingwei and Qiu, Zhaofan and Mei, Tao},
  booktitle={Proceedings of the IEEE/CVF international conference on computer vision},
  pages={18024--18033},
  year={2025}
}

@inproceedings{yu2025randomized,
  title={Randomized autoregressive visual generation},
  author={Yu, Qihang and He, Ju and Deng, Xueqing and Shen, Xiaohui and Chen, Liang-Chieh},
  booktitle={Proceedings of the IEEE/CVF International Conference on Computer Vision},
  pages={18431--18441},
  year={2025}
}

@article{sun2025speed,
  title={Speed always wins: A survey on efficient architectures for large language models},
  author={Sun, Weigao and Hu, Jiaxi and Zhou, Yucheng and Du, Jusen and Lan, Disen and Wang, Kexin and Zhu, Tong and Qu, Xiaoye and Zhang, Yu and Mo, Xiaoyu and others},
  journal={arXiv preprint arXiv:2508.09834},
  year={2025}
}

@misc{zhou2025medical,
  title={From Medical LLMs to Versatile Medical Agents: A Comprehensive Survey},
  author={Zhou, Yucheng and Zheng, Huan and Chen, Dubing and Yang, Hongji and Han, Wencheng and Shen, Jianbing},
  year={2025}
}

@misc{zhou2024less,
  title={Less is more: Vision representation compression for efficient video generation with large language models},
  author={Zhou, Yucheng and Zhang, Jihai and Chen, Guanjie and Shen, Jianbing and Cheng, Yu},
  year={2024},
  publisher={OpenReview}
}

@article{zhou2025draw,
  title={Draw ALL Your Imagine: A Holistic Benchmark and Agent Framework for Complex Instruction-based Image Generation},
  author={Zhou, Yucheng and Yuan, Jiahao and Wang, Qianning},
  journal={arXiv preprint arXiv:2505.24787},
  year={2025}
}
